%% file: main.tex
\documentclass[conference]{IEEEtran}
\IEEEoverridecommandlockouts

\usepackage{cite}
\usepackage{amsmath,amssymb,amsfonts}
\usepackage{algorithmic}
\usepackage{graphicx}
\usepackage{textcomp}
\usepackage{xcolor}
\def\BibTeX{{\rm B\kern-.05em{\sc i\kern-.025em b}\kern-.08em
    T\kern-.1667em\lower.7ex\hbox{E}\kern-.125emX}}

\input{package.tex}

\begin{document}
\input{content/0.title.tex}

\maketitle

\input{content/0.abstract.tex}

\input{content/1.Introduction.tex}
\input{content/2.0.Experimentsetup.tex}
\input{content/2.1.nclairAlgo.tex}
\input{content/2.2clairAlgo.tex}

\input{content/2.3learningaug.tex}
\input{content/3.Conclusion.tex}

\section*{Acknowledgment}
This research is supported by the Ministry of Education, Singapore, under its AcRF Tier 1 (Award RG15/25).

\bibliographystyle{IEEEtran}
\bibliography{bib/DVBP_Learning.bib, bib/DVBP.bib, bib/others.bib, bib/VBP.bib}

\input{content/app-CR}
\input{content/app-huawei}

\input{content/app-uniform_dist_error}

\end{document}

%% file: package.tex
\usepackage{amsmath}
\usepackage{amsfonts}
\usepackage{subcaption}
\usepackage{amsthm}
\usepackage[normalem]{ulem}
\usepackage{enumitem}
\usepackage{bbm}
\newtheorem{theorem}{Theorem}

\newtheorem{proposition}{Proposition}
\newcommand{\norminf}[1]{\left\| #1 \right\|_{\infty}}
\newcommand{\ceil}[1]{\left \lceil{ #1 }\right \rceil }
\newcommand{\OPT}{\mathrm{OPT}}
\DeclareMathOperator{\sspan}{span}
\allowdisplaybreaks[4]
\usepackage[table]{xcolor}
\usepackage{longtable}
\usepackage{url}

%% file: content/0.title.tex
\title{Evaluation of Dynamic Vector Bin Packing for Virtual Machine Placement}

\author{\IEEEauthorblockN{Zong Yu Lee and Xueyan Tang}
\IEEEauthorblockA{Nanyang Technological University, Singapore}
\IEEEauthorblockA{Email: zongyu001@e.ntu.edu.sg, asxytang@ntu.edu.sg}}

%% file: content/0.abstract.tex
\begin{abstract}
Virtual machine placement is a crucial challenge in cloud computing for efficiently utilizing physical machine resources in data centers. Virtual machine placement can be formulated as a \textit{MinUsageTime} Dynamic Vector Bin Packing (DVBP) problem, aiming to minimize the total usage time of the physical machines. This paper evaluates state-of-the-art \textit{MinUsageTime} DVBP algorithms in non-clairvoyant, clairvoyant and learning-augmented online settings, where item durations (virtual machine lifetimes) are unknown, known and predicted, respectively. Besides the algorithms taken from the literature, we also develop several new algorithms or enhancements. Empirical experimentation is carried out with real-world datasets of Microsoft Azure. The insights from the experimental results are discussed to explore the structures of algorithms and promising design elements that work well in practice.
\end{abstract}

%% file: content/1.Introduction.tex
\section{Introduction}
Cloud computing implements the delivery of computing resources to clients through the Internet by adopting a pay-as-you-go model \cite{overviewAWS}. When a client makes a request, the cloud system typically creates a virtual machine, which uses a specified amount of computing resources such as CPU, memory, and storage in a physical machine. There are usually many physical machines in a data center. The cloud system needs to select a physical machine with sufficient resources available to host a requested virtual machine \cite{overviewAWS}. This is known as \textit{virtual machine placement} \cite{overviewVMPlacement}. To serve more clients, multiple virtual machines often share the resources in a physical machine through virtualization. However, if the packing of virtual machines into the physical machines is inefficient, it can cause a waste of resources, known as fragmentation. Fragmentation may lead to resource over-provisioning. According to \cite{protean2020}, even a 1\% reduction in fragmentation could save a cost in the order of \$100M per year for large cloud providers. 

The virtual machine placement problem can be modelled as a \textit{MinUsageTime} Dynamic Vector Bin Packing (DVBP) problem. \textit{MinUsageTime} DVBP is a variation of the classical bin packing problem in which there are items (virtual machines) of vector sizes arriving and departing over time, and they are to be packed into bins (physical machines) such that the capacity constraints of bins are observed at any time. Unlike classical bin packing \cite{STOC_VBP, DBP1983}, the objective of \textit{MinUsageTime} DVBP is to minimize the total usage time of all bins, where the usage time of a bin is defined as the duration in which at least one item is placed in the bin (which indicates the time that the physical machine is used). Extensive research has been done to introduce online packing algorithms and theoretically analyze their competitive ratios (worst-case ratio of the algorithm performance over the optimal solution). However, existing work suffers from two major deficiencies. First, most studies considered the single-dimensional case only, where both item sizes and bin capacities are represented by scalar values \cite{clairdvbp2016,nextfit, clairdvbp2019,OnfirstFit2016,dvp2016ondemand,dbp2014,dbp2017,predictiondvbp2022,predictiondvbp2024,buchbinder2021online}. In contrast, virtual machines normally demand multiple types of resources (CPU, memory, storage, etc.) necessitating the consideration of the multi-dimensional case. Second, very little research evaluated the proposed algorithms empirically. The competitive ratio typically characterizes the algorithm performance in some corner cases. The practical performance of the packing algorithms for virtual machine placement remains unclear. To the best of our knowledge, only two studies investigated \textit{MinUsageTime} DVBP in the multi-dimensional case and conducted empirical comparisons \cite{murhekar2023dynamic, greedy}. Nevertheless, these evaluations used synthetic data only and not real-life datasets. Moreover, they did not explore the learning-augmented setting (where some item characteristics are predicted and they are prone to errors) and algorithms.

In this paper, we aim to bring the theory to practice by evaluating a comprehensive collection of packing algorithms using real-world datasets of Microsoft Azure. Besides the state-of-the-art algorithms taken from the literature, we also propose several new algorithms or enhancements and include them in the evaluation. We carry out extensive empirical evaluations of \textit{MinUsageTime} DVBP algorithms in a variety of different settings. The evaluation results provide insights into the structures of algorithms that work well in practice. Our contributions are summarized as follows.
\begin{itemize}
\item We develop several new algorithms or enhancements, including Round Robin Next Fit, Nearest Remaining Time, modified Hybrid, modified RCP and PPE. We also conduct competitive analysis for the new algorithms and improve the analysis of some existing algorithms.
\item Our evaluation shows that Any Fit (never open a new bin for an incoming item if it can fit in any current open bin) is a promising feature that can improve empirical packing efficiency if incorporated into the algorithm design. 
\item We show that categorizing items by their departure times outperforms that by their durations in the clairvoyant setting, but the former is less robust against prediction errors than the latter in the learning-augmented setting.
\item Our evaluation also shows that dynamic categorization of bins performs better than static categorization.
\end{itemize}

The rest of this paper is organized as follows. 
Section \ref{sec:problem} presents a formal definition of the \textit{MinUsageTime} DVBP problem. Section \ref{sec:setup} discusses the experimental setup. Sections \ref{sec:nonclairvoyant}, \ref{sec:clairvoyant} and \ref{sec:learningaugmented} introduce and evaluate algorithms in the non-clairvoyant, clairvoyant and learning-augmented online settings respectively. Finally, Section \ref{sec:conclusion} concludes the paper.

\section{Problem Definition and Preliminaries}
\label{sec:problem}
The inputs to the \textit{MinUsageTime} Dynamic Vector Bin Packing (DVBP) problem include a set of items $\mathcal{R}$, where each item $r \in \mathcal{R}$ is characterized by a $d$-dimensional vector $s(r)$ describing the item size and a time interval $I(r)$. The $d$-dimensional vector $s(r)$ represents the sizes of the item along $d$ different dimensions. The time interval $I(r) = [I(r)^-, I(r)^+)$ represents the active interval of the item, where $I(r)^-$ denotes the arrival time and $I(r)^+$ denotes the departure time of the item. The interval length $I(r)^+ - I(r)^-$ is called the duration of the item. The item is said \textit{active} during $I(r)$. 

There is a sufficiently large pool $\mathcal{P}$ of bins available to pack items. All the bins have the same capacity. The capacity of a bin is also represented by a $d$-dimensional vector. Without loss of generality, we assume that the capacity of a bin is given by $\langle 1, 1, \dots, 1 \rangle$, that is, all dimensions are normalized. Accordingly, it can be assumed that the size of an item along each dimension is no larger than $1$. 

In the context of virtual machine placement, each item models a virtual machine (VM). The $d$-dimensional size vector models the resource demands of the VM for $d$ different types of resources (CPU, memory, storage, etc.). The arrival and departure of the item model the creation and termination of the VM respectively. Each bin corresponds to a physical machine (PM). The $d$-dimensional capacity vector models the resource capacities of the PM for the $d$ types of resources.

A packing algorithm places each item into a bin while ensuring that at any time, the total size of the items placed in a bin does not exceed its capacity along all dimensions. The \textit{usage time} of a bin is the length of time in which there is at least one item placed in the bin. The accumulated bin usage time is the sum of the usage times of all the bins. From another perspective, the accumulated bin usage time can be considered as the integration of the number of bins used at all time instants over the time horizon.

The objective of the \textit{MinUsageTime} DVBP problem is to design a packing algorithm that minimizes the accumulated bin usage time. The key to optimizing the bin usage time is to increase the packing efficiency and hence the capacity utilization of bins. Inefficient placement of items would result in fragmentation and unnecessary use of additional bins. The total bin usage time of an optimal packing solution for a set of items $\mathcal{R}$ has the following lower bound \cite{murhekar2023dynamic}:
\begin{equation}
\int_{-\infty}^{\infty}{\ceil{\norminf{\sum_{r \in \mathcal{R} \text{ and } t \in I(r)}{s(r)}}}}\, \mathrm{d}t, \label{eq:lowerbound} \\
\end{equation}
since the number of bins used at any time $t$ must be at least the $\ell_\infty$-norm (the maximum value across all dimensions) of the aggregate size vector of all the active items at $t$.

We focus on the online setting, where items are released at their respective arrival times. Each item must be placed into a bin immediately when it arrives, without any information of the items that will arrive in the future. The online setting can further be divided into non-clairvoyant and clairvoyant settings depending on whether  the departure time of an item is known at its arrival. In addition, a new setting known as the learning-augmented setting \cite{kumar2018nips,lykouris2018icml} has emerged recently, which will be explained later. Theoretically, the competitive ratio of an online algorithm refers to the worst-case ratio between a solution constructed by the algorithm and an optimal solution \cite{Borodin1998OnlineCA}. 

%% file: content/2.0.Experimentsetup.tex
\section{Experiment Setup}
\label{sec:setup}

We use two publicly available datasets from Microsoft Azure and Huawei Cloud to empirically evaluate \textit{MinUsageTime} DVBP.

Microsoft Azure \cite{protean2020} releases part of the workload on Microsoft's Azure Compute over a 14-day period (\url{https://github.com/Azure/AzurePublicDataset/blob/master/AzureTracesForPacking2020.md}).

The dataset consists of a VMType table and a VMRequest table. There are different virtual machine (VM) types with distinct resource demands and different physical machine (PM) types with distinct resource capacities. Each record in the VMType table specifies the resource usage of a virtual machine type on a physical machine type. The fields of the VMType schema include:
\begin{itemize}
    \item virtual machine (VM) type,
    \item physical machine (PM) type,
    \item core, memory, HDD, SSD, NIC (fractional usage of core, memory, hard disk drive, solid-state drive, network bandwidth by the VM on the PM),    
\end{itemize}
Each record in the VMRequest table specifies a virtual machine request. The fields of the VMRequest schema include:
\begin{itemize}
    \item virtual machine (VM) type requested,
    \item starttime (the starting time of the VM),
    \item endtime (the termination time of the VM).
\end{itemize}
We clean and convert the dataset into \textit{MinUsageTime} DVBP instances as follows.

    \textbf{Records.} The VMType table has 488 different VM types, 35 different PM types, and a total of 4619 records (not all records of distinct VM and PM type pairs are available because some pairs are not compatible with each other \cite{protean2020}). The VMRequest table has 5,559,800 requests. 

    \textbf{Time Information.} The VMRequest table contains complete workload in the 14-day period of data collection. In the VMRequest table, there are records where the \textit{endtime} field has null values, which means that the VM is not terminated by the end of data collection. We removed them because their VM lifetimes are unknown. Meanwhile, there are records where the \textit{starttime} field has negative values, which means that the VM is created before the start of data collection; and there are also records where the \textit{endtime} field has values greater than 14 days, which indicates that the data collector continues to monitor the created VM (at most until 90 days) and the VM is terminated within the monitoring period. The above requests with VM lifetimes partially falling outside the 14-day period are not included in the evaluation because we do not have the complete workload before and after the 14-day period. After the aforesaid records are removed, there are a total of 4,551,058 requests (about 82\%) left whose VM lifetimes completely fall in the 14-day period. Figure \ref{fig:dist} shows the distribution of these VM lifetimes and Figure \ref{fig:qqplot} presents the normal quantile-quantile plot for the logarithm of VM lifetimes, which shows that the VM lifetimes almost follow a log-normal distribution.

\begin{figure}[t]
\centering
\begin{subfigure}{\linewidth}
  \centering
  \includegraphics[width=0.8\linewidth]{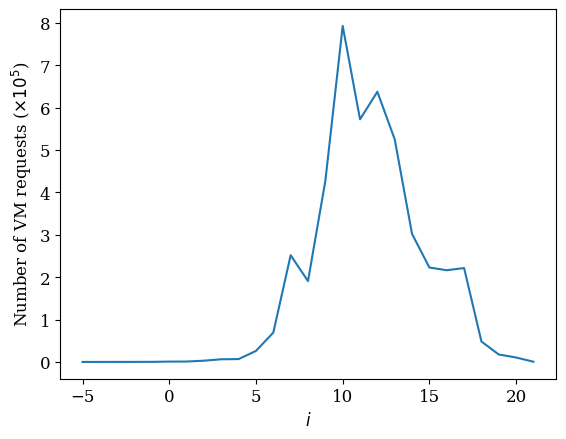}
  \caption{Number of VM requests with lifetime in $[2^{i-1},2^i)$ seconds}
  \label{fig:dist}
\end{subfigure}

\begin{subfigure}{\linewidth}
  \centering
  \includegraphics[width=0.8\linewidth]{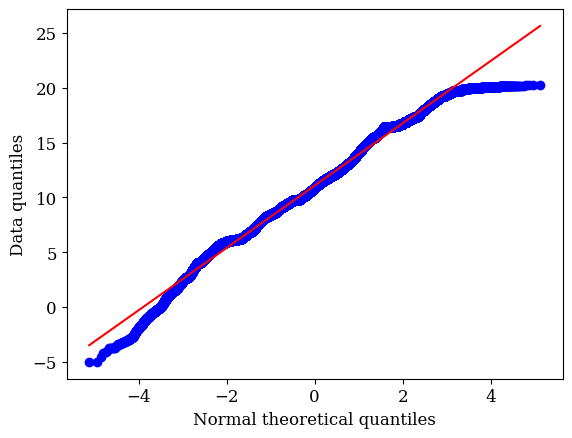}
  \caption{Q-Q plot of the log lifetime of all VM requests}
  \label{fig:qqplot}
\end{subfigure}
\caption{Exploratory analysis of VM lifetimes}
\label{fig:subplot2}
\end{figure}

    \textbf{Instances.} Our evaluation assumes that there is a sufficiently large number of physical machines available and all the physical machines are of the same type with the same resource capacity. Thus, 35 possible \textit{MinUsageTime} DVBP instances can be created from the dataset, one for each PM type. Given a PM type, the VM requests for VM types compatible with the PM type are extracted to form the input to the instance. Among 35 PM types, 2 of them do not have any VM requests because there is no request for compatible VM types. In addition, some VM types are exact duplications of other VM types in terms of compatibility and resource usage with respect to PM types. Therefore, some instances created are exactly the same. After removing the two empty instances and the redundant instances (there are 4 of them), we end up having 29 distinct instances. We further exclude one particular instance in our evaluation because it is trivial (due to the low number and resource demands of VM requests in this instance, all the active VMs can fit into one PM at any time).

    \textbf{Resource Dimensions.} In the VMType records, there are five resource dimensions ($d = 5$): core, memory, HDD, SSD and NIC. Some VMType records have null values or zero values in the HDD field. We assume that null values mean zero (i.e., the VM does not use any resource of the corresponding dimension). For some PM types, all the VM types compatible with the PM type have zero/null values in the HDD field. As a result, the resource dimension of HDD does not affect VM placement and thus can be ignored (the PM type probably does not have HDD resources at all). Hence, in such cases, only four resource dimensions ($d = 4$) are considered in our evaluation: core, memory, SSD and NIC. 
    
We run different packing algorithms for each of the 28 instances to obtain the total usage time of all bins. We compute the ratio of the total usage time produced by an algorithm over the best-known lower bound of the optimal solution as defined in Equation \eqref{eq:lowerbound}. The ratio is intended to capture the normalized performance of an algorithm compared with the optimal solution for a particular instance. We refer to it as the \textit{performance ratio}. The lower bound is used instead of the exact optimal solution because \textit{MinUsageTime} DVBP is NP-hard, and it is computationally expensive to derive the exact optimal solution. Obviously, the actual ratio between the total bin usage times of an algorithm and the optimal solution is never higher than our performance ratio. We present the box plot of the performance ratios to examine the performance of an algorithm across all instances. We use the common box plot showing the first quartile (25th percentile) and the third quartile (75th percentile) as a box, and the boundaries of the whiskers based on 1.5 times the interquartile range (which is defined as the distance between the third and first quartiles). 

Huawei Cloud also releases a VM placement dataset \cite{sheng2021vmagent} (\url{https://github.com/huaweicloud/VM-placement-dataset}) that has a schema largely similar to the Microsoft Azure dataset but has a much smaller size (116,313 VM requests after cleaning and only two resource dimensions: CPU and memory). The evaluation results using the Microsoft Azure and Huawei Cloud datasets show similar trends. Due to space limitations, we shall focus on presenting the results of the Microsoft Azure dataset in this paper. The results of the Huawei Cloud dataset are provided in Appendix \ref{app:huawei}.

%% file: content/2.1.nclairAlgo.tex
\section{Non-clairvoyant Algorithms}
\label{sec:nonclairvoyant}
This section introduces and evaluates algorithms in the non-clairvoyant setting -- the departure time of each item is not known at its arrival and hence cannot be used in making packing decisions. Since the packing algorithm is blind to item departure times, it cannot avoid mixing items with short and long durations in the same bin. Indeed, it has been proved that no deterministic packing algorithm can achieve a competitive ratio lower than $\mu$ (the max/min item duration ratio) in the single-dimensional case \cite{dbp2014}. 

In the packing process, a bin is said \textit{opened} when it receives the first item. When all the items in a bin depart, the bin is said \textit{closed}. Any Fit is an important feature of many online packing algorithms \cite{dbp2014}. It means that the algorithm never opens a new bin for an incoming item if the item can fit into any current open bin. Any Fit is generally good for \textit{MinUsageTime} DVBP because item placement will always try to make use of open bins that do not incur additional bin usage time immediately.
Theoretically, it has been shown that no Any Fit packing algorithm can achieve a competitive ratio lower than $(\mu+1)d$ in the $d$-dimensional case \cite{murhekar2023dynamic}.

\subsection{First Fit and Most Recently Used} 
First Fit \cite{dbp2014,dbp2017} and Most Recently Used (MRU) \cite{nextfit, murhekar2023dynamic} are Any Fit algorithms. These algorithms maintain different orderings of the open bins. Whenever an item arrives, the algorithm goes through the particular order of the open bins and places the item into the first bin that has enough capacity available for it in all resource dimensions. 

First Fit arranges the bins in ascending order of open time, so an incoming item is placed in the earliest-opened bin. Most Recently Used arranges the bins in descending order of last access time, where a bin is said accessed when a new item is placed in the bin. In this way, receiving a new item moves the bin to the front of the ordering -- (some literature named this algorithm Move to Front \cite{nextfit, murhekar2023dynamic}). First Fit can achieve a competitive ratio of $(\mu + 2)d + 1$ \cite{dbp2017, murhekar2023dynamic}, which in the single-dimensional case is pretty close to the lower bound of $\mu$ on the competitiveness of any deterministic algorithm \cite{dbp2014}. MRU can achieve a competitive ratio of $(2\mu + 1)d + 1$ \cite{murhekar2023dynamic}.

\subsection{Next Fit and Round Robin Next Fit}
In Next Fit \cite{nextfit}, only one open bin accepts new item placement at any time. Once an incoming item cannot fit into this bin, this bin will no longer accept any further item, and a new bin will be opened to receive the incoming item and subsequently arriving items. 

Next Fit originates in a static setting where items never depart. The main feature of Next Fit is to prevent mixing items arriving far apart in the same bin. In a dynamic setting, Next Fit is not likely to perform very well because all the open bins except the most recently opened one are not allowed to receive new items. An earlier-opened bin may be able to accommodate new items after some items placed in it depart. Next Fit is not an Any Fit algorithm as a new bin may be opened for an incoming item even if the item can fit into an earlier-opened bin. We propose Round-Robin Next Fit as a new algorithm that converts the concept of Next Fit into an Any Fit algorithm, which is more suitable for a dynamic setting. When an incoming item cannot fit into the open bin currently under consideration, instead of opening a new bin, we can loop around to search for any possible open bin that can accommodate the item. That is, we maintain the open bins in ascending order of open time (same as First Fit) and keep a flag indicating the open bin currently under consideration. For each incoming item, the flag initially points to the bin receiving the last arrived item and the flag will be moved around in a round-robin manner to search for an open bin for placing the incoming item. Only when no open bin can accommodate the incoming item will a new bin be opened. Next Fit achieves a competitive ratio of $2\mu d + 1$ \cite{nextfit, murhekar2023dynamic}. We prove that Round-Robin Next Fit achieves a competitive ratio of $(2\mu+1) d+1$. The main idea is as follows. The usage time of each bin can be partitioned into alternating periods of receiving items (the flag points to the bin) and not receiving items (the flag does not point to the bin). At any time, there is only one bin in the receiving period. The non-receiving period of any other open bin $b$ is initiated by an item $r$ placed into another bin (i.e., $r$ causes the flag to move from $b$ to the latter bin). By the algorithm definition, the total size of items in the non-receiving period of $b$ and the item $r$ must exceed the bin capacity. If we duplicate the item $r$ and extend its duration by $\mu$, and extend the durations of the items in $b$ by $\mu + 1$, all these items would be active throughout the non-receiving period of $b$. Hence, the number of open bins is bounded by the total size of active items in the original and duplicate sets, which in turn is bounded by the optimal packing solution. The full proof is deferred to Appendix \ref{sec:rrnextfitanalysis}. We also construct an adversary example to show that the competitive ratio of Round Robin Next Fit is at least $2 \mu d$ in Appendix \ref{sec:rrnextfitanalysis}.

\subsection{Best Fit} \label{sec:bestfit}

The guiding principle of Best Fit \cite{dbp2014} is to reduce the wastage of bin capacity. Among all open bins, an incoming item will be placed into the open bin whose available capacity ``best'' fits the size of the item, i.e., the open bin that will leave the least remaining capacity after accepting the item. In the single-dimensional case, this simply means to place the item into the most-loaded open bin. In the multi-dimensional case, however, the item size and bin capacity are described by vectors due to multiple resource types. Hence, there are many different ways to define what is the ``least'' vector. We adopt the definition of $\ell_p$-norm to characterize vectors. Let $\vec{a} = \langle a_1, a_2, \dots, a_d \rangle$ be the size of an incoming item, and $\vec{b} = \langle b_1, b_2, \dots, b_d \rangle$ be the available capacity of an open bin. Then, the remaining capacity of the bin after accepting the item is given by $\vec{b} - \vec{a} = \langle b_1 - a_1, b_2 - a_2, \dots, b_d - a_d \rangle$. We explore three common $\ell_p$-norms to compute the fit score: $\ell_1$-norm (sum of available capacity over all dimensions: $\sum_{1 \leq i \leq d}{ (b_i - a_i)}$), $\ell_2$-norm (square root of the sum of squares: $\sqrt{\sum_{1 \leq i \leq d}{(b_i - a_i) ^ 2}}$), and $\ell_\infty$-norm (maximum available capacity across all dimensions: $\max_{1 \leq i \leq d}{(b_i - a_i)}$). The incoming item will be placed in the open bin with the lowest fit score that can accommodate it. It has been proved that Best Fit cannot achieve any bounded competitive ratio, even in the single-dimensional case \cite{dbp2014}.

\textbf{Experimental Comparison:} Figure \ref{fig:bestworstbox} compares the performance ratios of different norms for Best Fit. In general, $\ell_\infty$-norm works the best for Best Fit. As seen from Figure~\ref{fig:bestworstbox}, $\ell_\infty$-norm produces the lowest mean performance ratio. A possible reason is that $\ell_\infty$-norm captures the dimension with the maximum available capacity and acts as an upper bound of the available capacity in other dimensions. As such, $\ell_\infty$-norm is a good representation of the available capacity of a bin. In what follows, we shall focus on using the $\ell_\infty$-norm in our evaluation. 

\begin{figure}
\centering
\includegraphics[width=0.31\textwidth]{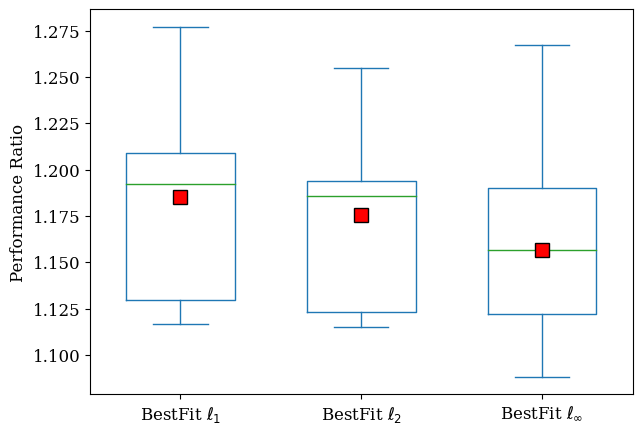}
\caption{\label{fig:bestworstbox} Performance of Best Fit}
\end{figure}

\subsection{Experimental Comparison in Non-clairvoyant Setting}

Figure \ref{fig:nclairBox} presents the box plot for various algorithms in the non-clairvoyant setting. To facilitate the comparison among better-performing algorithms, we limit the cap of the vertical axis (performance ratio) to $2.5$, which cuts off the box/outliers for some algorithms.\footnote{Similar caps on the vertical axis are also imposed in subsequent figures.} We could draw the following insights:

    \textbf{Best Performer}: The First Fit algorithm produces the lowest mean performance ratio among all the algorithms tested, which concurs with that it has the lowest competitive ratio. 

    \textbf{Any Fit Feature}: It is beneficial to enhance a non-clairvoyant algorithm with the Any Fit feature. Round Robin Next Fit considerably improves the performance ratio over Next Fit, though they have very similar competitive ratios. This shows that by utilizing open bins as much as possible, the Any Fit feature can help improve the packing efficiency dramatically.
    
    \textbf{Best Fit and MRU:} Best Fit performs close to First Fit, because minimizing the leftover capacity in the open bins makes item placement concentrate in fewer bins. Most Recently Used (MRU) performs close to Round Robin Next Fit. We note that earlier studies \cite{murhekar2023dynamic, greedy} showed that MRU is empirically the best in their evaluations with synthetic data, but this is not the case in our evaluation with real-world data of VM placement. This could be due to different distributions of item sizes and durations (they were generated from uniform distributions in \cite{murhekar2023dynamic} and \cite{greedy}).

\begin{figure}
\centering
\includegraphics[width=.47\textwidth]{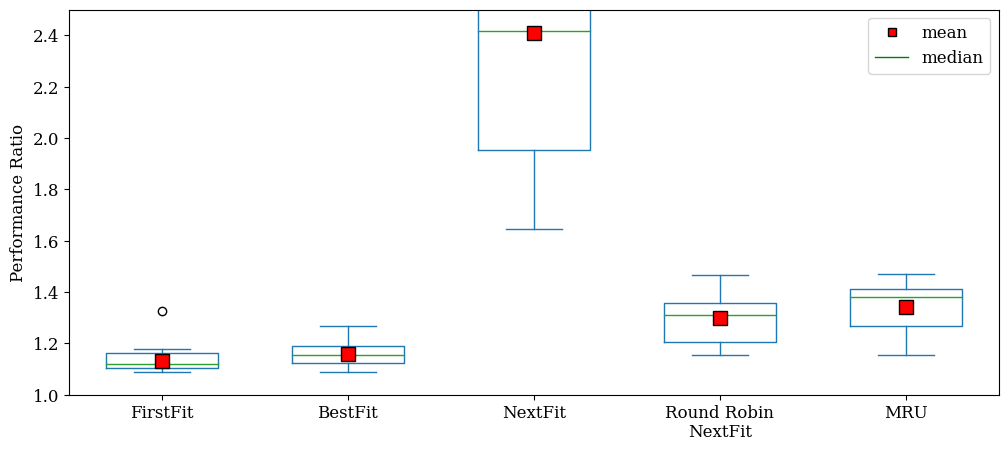}
\caption{\label{fig:nclairBox} Performance of non-clairvoyant algorithms}
\end{figure}

%% file: content/2.2clairAlgo.tex
\section{Clairvoyant Algorithms}
\label{sec:clairvoyant}
This section presents and evaluates algorithms in the clairvoyant setting -- the departure time (or equivalently the duration) of each item is known upon its arrival. Algorithms in the clairvoyant setting usually utilize the departure time or item duration information to categorize the items and bins. The items of each category are placed in the bins of the same category in a First Fit manner unless stated otherwise. Note, however, that algorithms of this paradigm are not Any Fit algorithms if there are multiple categories.
This is nevertheless necessary for packing algorithms to achieve competitive ratios of $o(\mu d)$, since the lower bound of $(\mu+1)d$ on the competitiveness of Any Fit algorithms still applies in the clairvoyant setting. In what follows, we introduce the details of each algorithm and its intuition. We start with algorithms categorizing items based on departure times, followed by algorithms conducting categorization based on item durations. 

\subsection{Classify By Departure Time}

Wastage of bin capacity is likely to occur when some items in a bin depart while other items in the bin do not depart. With the knowledge of item departure times, an intuitive approach to reduce the bin usage time is to place items with similar departure times in the same bin \cite{clairdvbp2016}. In this way, the items in a bin would depart at around the same time, after which the bin can be closed timely. As a result, the bin usage time can be saved. The Classify-By-Departure-Time algorithm operates with a parameter $\rho$ defined as the width of the departure time range. It partitions the time horizon into disjoint intervals of length $\rho$ and puts all the items with departure times in each interval into one distinct category. The algorithm can achieve a competitive ratio of $O(\sqrt{\mu})$ in the single-dimensional case by setting $\rho$ to the geometric mean of the maximum and minimum item durations if they are known upfront \cite{clairdvbp2016}. There is no competitive analysis for the multi-dimensional case in the existing literature. Therefore, we conduct experiments for different $\rho$ values.

\textbf{Experimental Comparison}: As seen from the box plot in Figure \ref{fig:depart}, the median and mean of performance ratios exhibit a U-shaped trend. When $\rho$ is small, the algorithm creates many categories, which potentially degrades the packing efficiency because items of different categories are packed into separate bins. On the other hand, when $\rho$ is large, the algorithm creates fewer categories, but it degenerates towards the First Fit algorithm which does not exploit any departure time information. As such, there is a trade-off between overusing and underusing the departure time information of items in the packing. For the Microsoft Azure dataset, setting $\rho$ to $0.25$ day appears to produce the best performance ratios. Thus, we shall use this $\rho$ setting in subsequent comparisons. 

\begin{figure}
\centering
\includegraphics[width=0.48\textwidth]{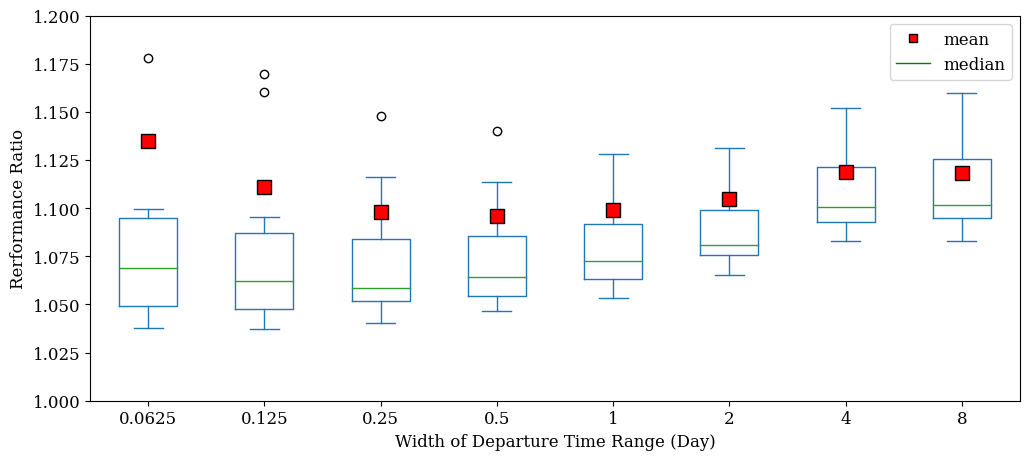}
\caption{\label{fig:depart} Performance of Classify-By-Departure-Time}
\end{figure}

\subsection{Nearest Remaining Time (NRT)}

We propose an algorithm called Nearest Remaining Time that extends the idea of classification by departure time. The central intuition is again that items with similar departure times should be placed in the same bin. However, Classify-By-Departure-Time is not an Any Fit algorithm. To improve the algorithm design, we integrate the Any Fit concept by dynamically defining the departure time range. Note that the latest departure time of all items currently placed in a bin indicates the bin closing time (referred to as the \textit{indicated closing time} hereafter). Among all the open bins with enough available capacity, we can place an incoming item into the bin that has the indicated closing time nearest to the departure time of the item. From an alternative perspective, if we arrange the bins in ascending order of indicated closing time, this strategy implies that each bin will receive items with departure times in the range from the midpoint between the previous bin and this bin to the midpoint between this bin and the next bin. Therefore, it dynamically defines the range of accepting departure times for each bin. 

To enhance the algorithm design, we can further distinguish between case (a) where the bin’s indicated closing time is no earlier than the item departure time and case (b) where the bin’s indicated closing time is earlier than the item departure time. In case (a), placing the item into the bin does not increase the bin usage time, whereas in case (b), it extends the bin's indicated closing time and hence increases the bin usage time. Thus, to reduce the bin usage time, priorities can be given to the bins of case (a) over the bins of case (b) in item placement. Therefore, we evaluate two versions of the Nearest Remaining Time (NRT) strategy as follows:
\begin{itemize}
    \item \textbf{Standard NRT} -- Consider case (a) and case (b) together and choose the bin with indicated closing time nearest to the departure time of an item for placement. 
    \item \textbf{Prioritized NRT} -- An item is placed in a bin of case (b) only if the item cannot fit into any bin of case (a). For each case, choose the bin with indicated closing time nearest to the item departure time for placement.
\end{itemize}

Standard NRT does not have a bounded competitive ratio. In fact, an adversary can force Standard NRT to maintain many open bins indefinitely, whereas the optimal packing solution needs just one bin. First, the adversary releases $n$ items each of size $\frac{1}{2} + \epsilon$ ($\epsilon > 0$) to open $n$ bins $b_1, b_2, \dots, b_n$ with increasing indicated closing times forming an arithmetic sequence. Then, the adversary releases items in rounds. In each round, the adversary releases $n$ items each of size $\epsilon$, where the departure time of the $i$-th item falls between the indicated closing times of $b_{n-i+1}$ and $b_{n-i}$ and is slightly closer to the former than the latter. Standard NRT will place the $i$-th item in bin $b_{n-i+1}$ and extend $b_{n-i+1}$'s indicated closing time. After each round, the indicated closing times of the bins again form an arithmetic sequence. As $\epsilon \rightarrow 0$, all the active items can be placed into just one bin after the initial items depart. Hence, an optimal packing solution will only need one bin to accommodate all the items. Thus, the competitive ratio is not bounded as $n \rightarrow \infty$.

We prove that Prioritized NRT achieves a competitive ratio of $(\mu + 2)d + 1$, which is rather close to the lower bound of $(\mu + 1)d$ on the competitiveness of Any Fit algorithms. The main idea is as follows. At any time $t$, we sort the open bins by the earliest time instant when the bin has its indicated closing time greater than $t$ due to a new item placed in it (call it the triggering item). By the algorithm definition, for any $j$, the total size of active items in the $(j-1)$-th bin and the triggering item of the $j$-th bin must exceed the bin capacity. If we duplicate the items and extend their durations by $\mu + 1$, all these items would be active at time $t$. Hence, the number of open bins at $t$ is bounded by the total size of items active at $t$ in the original and duplicate sets, which in turn is bounded by the optimal packing solution. The full proof is deferred to Appendix \ref{sec:greedyandnrtanalysis} due to space limitations.

\textbf{Experimental Comparison}: Figure \ref{fig:CRTbox} compares the performance ratios of Standard NRT and Prioritized NRT. It can be seen that Prioritized NRT works significantly better than Standard NRT. Though they both have the Any Fit feature, Prioritized NRT prevents the bin closing time from increasing as far as possible and thus avoids additional bin usage time. Therefore, we shall focus on Prioritized NRT in our evaluation. 

\begin{figure}
\centering
\includegraphics[width=.32\textwidth]{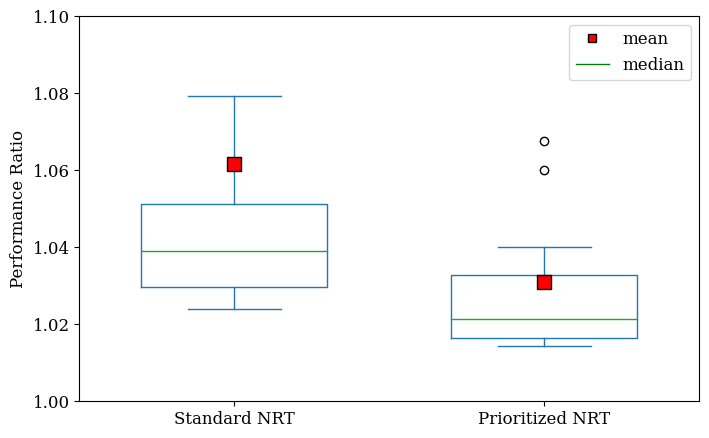}
\caption{\label{fig:CRTbox} Performance of NRT} 
\end{figure}

\subsection{Greedy}
The Greedy algorithm \cite{greedy} also makes placement decisions based on the relation between the bin's indicated closing time and the item's departure time. Different from NRT, it places an incoming item into the bin with the latest closing time, among all the open bins that can accommodate the item. If the item's departure time is later than the indicated closing time of the chosen bin, the bin will have its indicated closing time extended to the item's departure time after placement. Li \textit{et al.} \cite{greedy} proved that the Greedy algorithm can achieve a competitive ratio of $3\mu d + 1$ in the $d$-dimensional case. We further improve its competitive analysis in Appendix \ref{sec:greedyandnrtanalysis} and show that the Greedy algorithm achieves a competitive ratio of $(\mu + 2)d + 1$, following the same proof idea as Prioritized NRT.

\subsection{Classify By Duration}
Another strategy is to place items with similar durations in the same bin \cite{clairdvbp2016}. This is motivated by the observation that in the non-clairvoyant setting, the competitiveness of online packing is limited by the variation of item durations. By reducing the variation of item durations in the same bin, it is possible to improve the performance of the packing algorithm. Specifically, items are classified into categories such that the max/min item duration ratio for each category is a given constant $\beta$. Each category $i$ includes all the items with durations in the range $[\beta^{i-1}, \beta^i)$ for an integer $i$ (where $i$ can be negative, zero or positive). If $\mu$ is not known in advance, the Classify-By-Duration algorithm achieves a competitive ratio of $O(\log\mu)$ in the single-dimensional case \cite{clairdvbp2016}.

\textbf{Experimental Comparison}: We test different $\beta$ values. As shown by the box plot in Figure \ref{fig:duration}, the median and mean of performance ratios generally have a U-shaped trend. Similar to Classify-By-Departure-Time, when $\beta$ is small, the algorithm creates many categories, which potentially downgrades the packing efficiency; when $\beta$ is large, the algorithm creates fewer categories and approaches the First Fit algorithm. For the Microsoft Azure dataset, it appears that when $\beta = 2$, the median performance ratio is the lowest. Hence, in what follows, we shall set $\beta = 2$ in our evaluation. 

\begin{figure}
\centering
\includegraphics[width=0.48\textwidth]{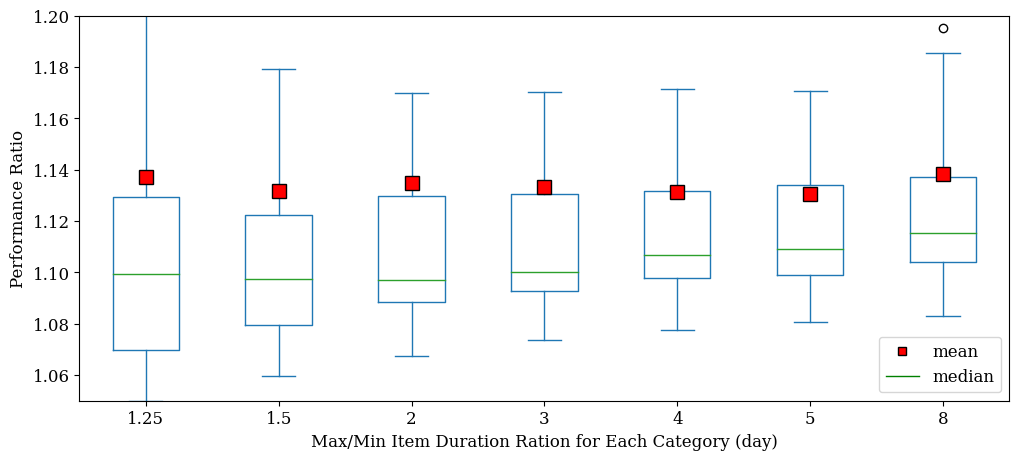}
\caption{\label{fig:duration} Performance of Classify-By-Duration}
\end{figure}

\subsection{Hybrid Algorithm}

The Hybrid algorithm (HA) \cite{clairdvbp2019} combines the First Fit and Classify By Duration strategies. It categorizes items based on their durations and arrival times. Each category $(i,c)$ includes all the items with durations in the range $[2^{i-1}, 2^i)$ for an integer $i$ and arrival times falling in the interval $[(c - 1) \cdot 2^i, c \cdot 2^i)$ for an integer $c$. The Hybrid algorithm defines two types of bins -- general bins and category bins. General bins can receive items of any categories, while each category bin can receive items of a particular category only. Items of each category $(i,c)$ are placed into the general bins when their total size in the general bins is no larger than a threshold $\frac{1}{2\sqrt{i}}$, and placed into the category bins when their total size in the general bins exceeds the threshold. The First Fit strategy is adopted for item placement among the general bins as well as among the category bins of the same category. The intuition behind the Hybrid algorithm is to implement the Classify By Duration strategy only when there is sufficient load of each category and prevent items from being unnecessarily distributed to many separate bins. Azar and Vainstein \cite{clairdvbp2019} proved that the Hybrid algorithm achieves a competitive ratio of $O(\sqrt{\log{\mu}})$, which is asymptotically optimal (best achievable) in the clairvoyant setting. The original Hybrid algorithm is for the single-dimensional case. To adapt it to the multi-dimensional case, we define the total size of active category-$(i,c)$ items in the general bins as the $\ell_{\infty}$-norm of their aggregate size vector, i.e., $ \norminf{\sum_{r \in \mathcal{R}^*}{s(r)}}$, where $\mathcal{R}^*$ is the set of active category-$(i,c)$ items placed in the general bins. We prove that such adaptation allows the Hybrid algorithm to achieve a competitive ratio of $O(d\sqrt{\log\mu})$. Due to space limitations, we defer the full proof to Appendix \ref{sec:reducedhybridanalysis}. We also note that the original Hybrid algorithm assumes the minimum item duration is at least $1$, so that $i$ in the category index is always a positive integer (hence $\sqrt{i}$ is well-defined). In our implementation, for each instance created from the dataset, we find the minimum item duration and the power-of-2 range $[2^{z-1}, 2^z)$ including this minimum item duration. We set $i=1$ for this range $[2^{z-1}, 2^z)$ and set $i=j-z+1$ for each subsequent range $[2^{j-1}, 2^j)$ where $j>z$ accordingly. 

Recently, Liu and Tang \cite{predictiondvbp2022} found that item categorization can be simplified while maintaining the asymptotic optimal $O(\sqrt{\log{\mu}})$ competitiveness of packing. It is sufficient to use a more coarse-grained item categorization: define categories based on item durations only. That is, each category $i$ includes all the items with durations in the range $[2^{i-1}, 2^i)$ for an integer $i$. This can significantly reduce the number of categories. The concepts of general bins and category-specific bins remain unchanged in the Hybrid framework. We also include this Reduced Hybrid algorithm in our evaluation. Again, we adapt it to the multi-dimensional case by using the $\ell_{\infty}$-norm to characterize the total size of active category-$i$ items in the general bins. We prove that the Reduced Hybrid algorithm also achieves a competitive ratio of $O(d\sqrt{\log\mu})$ in Appendix~\ref{sec:reducedhybridanalysis}.

Very recently, Li \textit{et al.} \cite{greedy} proposed a technique called \textit{direct-sum} to extend a packing algorithm in the single-dimensional case to the multi-dimensional case. The idea is to divide all items into $d$ classes based on the dimension in which the item has the largest size. Then, the packing algorithm is applied to each class of items separately, since the packing needs to look at only the dimension of largest item sizes. Li \textit{et al.} \cite{greedy} proved that such adaptation introduces a multiplicative factor of $d$ (the number of dimensions) to the competitive ratio. We also evaluate the Hybrid and Reduced Hybrid algorithms adapted using the direct-sum technique to the $d$-dimensional case, which also have a competitive ratio of $O(d\sqrt{\log\mu})$.

\begin{figure}
\centering
\includegraphics[width=.3\textwidth]{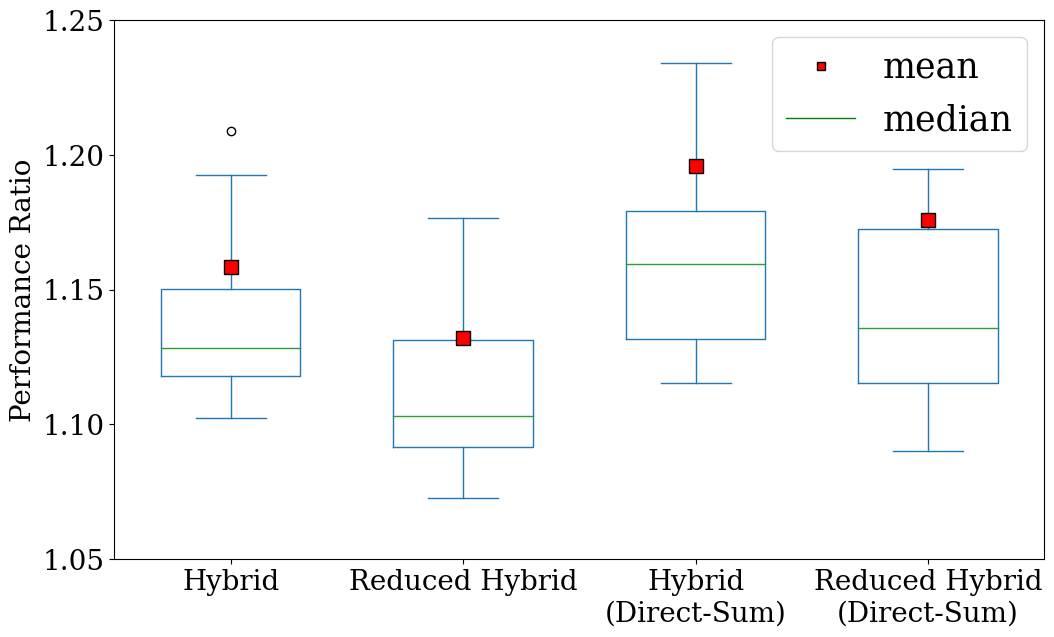}
\caption{\label{fig:HAbox} Performance of Hybrid algorithms}
\end{figure}

\textbf{Experimental Comparison}: Figure \ref{fig:HAbox} shows the experimental results. Comparing the left two algorithms, Reduced Hybrid considerably improves the packing efficiency over Hybrid. This is because eliminating the factor of item arrival times cuts down the number of categories which in turn leads to better utilization of bin capacity. Comparing our $\ell_\infty$-norm characterization (left two algorithms) with the direct-sum adaptation (right two algorithms), our adaptation is much more effective than direct-sum. This is because direct-sum essentially further increases the number of item categories by a factor of $d$, which downgrades the packing efficiency. In what follows, we shall focus on the Reduced Hybrid algorithm using our $\ell_\infty$-norm characterization in the evaluation. 

\subsection{Experimental Comparison in Clairvoyant Setting}

Figure \ref{fig:clairBox} presents the box plot for various algorithms in the clairvoyant setting. We could draw the following insights:
    
    \textbf{Best Performer}: Prioritized NRT performs very well and better than all the other algorithms. This demonstrates the effectiveness of matching item departure times with bin closing times and prioritizing bins that do not need to extend their closing times upon accepting items. 
        
    \textbf{Departure Time vs. Item Duration Information}: Classify-By-Departure-Time, Prioritized NRT and Greedy utilize the departure time information of items in the placement, whereas Classify-By-Duration and Reduced Hybrid utilize the item duration information in the placement. From Figure \ref{fig:clairBox}, algorithms using departure time information tend to perform better than algorithms using item duration information. While categorizing items based on duration can optimize the asymptotic competitive ratios from the theoretical perspective, it is practically less successful than leveraging the departure time information in reducing the total bin usage time. Aligning the departure times of items in the same bin allows the bin to be closed timely, thereby reducing the wastage of bin capacity effectively. 
        
    \textbf{Any Fit Feature}: Prioritized NRT and Greedy are Any Fit algorithms and they perform better than others which are not Any Fit algorithms. Even though Any Fit algorithms are never better than $O(\mu d)$ in competitiveness, empirically Any Fit is still promising. A good algorithm design is to integrate Any Fit with the departure time information of items. 
    
\begin{figure}
\centering
\includegraphics[width=.43\textwidth]{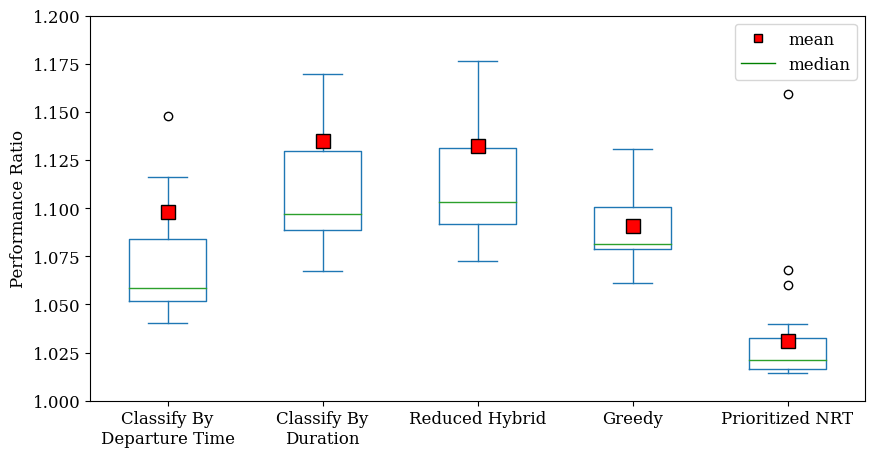}
\caption{Performance of clairvoyant algorithms}
\label{fig:clairBox}
\end{figure}

%% file: content/2.3learningaug.tex
\section{Learning-augmented Algorithms}
\label{sec:learningaugmented}
In practical applications, it is usually difficult to know the exact item durations (such as VM lifetimes or job lengths) at item arrivals. Item durations often need to be predicted based on history by using machine learning techniques. The prediction accuracy, however, varies widely with techniques and environments. Hence, it will be important to evaluate whether the packing efficiency is resilient to prediction errors. Making use of predictions to enhance online algorithm design has led to an arising research area called learning-augmented algorithms \cite{cacm2022}. This section presents and evaluates algorithms in the learning-augmented setting -- a predicted duration of each item is known upon its arrival, but it may not always be the same as the real duration of the item and can be arbitrarily different from the latter in the worst case. The goal is to make wise use of the predictions to improve the algorithm’s performance. 

\subsection{Robust and Consistent Packing (RCP) and Packing with Prediction Error (PPE)}

Liu and Tang \cite{predictiondvbp2022} developed a learning-augmented algorithm called Robust and Consistent Packing (RCP). It extends the Hybrid framework by defining four types of bins (large bins, general bins, base bin, category bins) and dynamically turning item categories on and off. The major extensions are as follows: 
\begin{itemize}
    \item Items are categorized according to their predicted durations rather than real durations;
    \item Each large item (of size larger than $1/2$) is placed in a separate large bin.\footnote{A ``large bin'' means a bin for placing large items only. It does not mean that the bin has a larger capacity than others.} 
    \item The threshold setting is simplified by applying a uniform and adjustable threshold value $1/{\sqrt{x}}$ for all the categories, where $x$ is the number of distinct categories seen by the algorithm so far (which typically grows as time goes on).
    \item Small items (of size no larger than $1/2$) go to the general bins first. For each category, when the total size of active items in the general bins exceeds the threshold $1/{\sqrt{x}}$, incoming items do not go to the category bins immediately. Instead, such incoming items of all categories go to a base bin. The base bin is a logical concept. At any time, there is at most one base bin. When the total size of items in the base bin exceeds $1/2$, the base bin is converted to a category bin of the category of a selected item in the bin, turning on that category. 
    \item When a category is on, incoming items of the category (exceeding the threshold) go to the category bins. When the total size of items in the category bins falls low, the category is turned off, after which incoming items of the category (exceeding the threshold) go to the base bin again. 
\end{itemize}
The base bin is designed to guarantee the resilience of the algorithm against terrible predictions. In the worst case, all incoming items can be very small in size and belong to different categories due to mispredictions. Directly applying the Hybrid framework will open many bins concurrently. In this case, the base bin is used to receive all items and avoid opening unnecessary bins. Liu and Tang \cite{predictiondvbp2022} proved that the RCP algorithm achieves $O(\mu)$ consistency (competitive ratio under perfect predictions) and $O(\sqrt{\log \mu})$ robustness (competitive ratio under terrible predictions), matching the asymptotic optimal competitive ratios in the clairvoyant and non-clairvoyant settings respectively. 

More recently, Liu and Tang \cite{predictiondvbp2024} further refined  the RCP algorithm and named it Packing with Prediction Error (PPE). The main idea is to monitor the prediction errors of departed items (note that the real duration of an item is revealed when it departs) and adjust the threshold value accordingly. Specifically, the threshold value is set to $\alpha/\sqrt{x}$, where $\alpha$ is an online estimate of the maximum prediction error using a guess-and-double approach. As a result, PPE achieves an asymptotic optimal competitive ratio of $O\big(\min\big\{\max\{\epsilon\cdot\sqrt{\log \mu}, \epsilon^2\}, \mu\big\}\big)$ over the entire spectrum of prediction errors, where $\epsilon$ is the maximum multiplicative prediction error among all items (the multiplicative prediction error is the ratio between the real duration and the predicted duration). We include both RCP and PPE in our evaluation. Similar to the Hybrid algorithm, we adapt RCP and PPE to the multi-dimensional case by using the $\ell_\infty$-norm to characterize the total size of a set of items.

In addition, we find that placing each large item (of size larger than $1/2$) in a separate large bin is not really necessary for RCP and PPE to achieve asymptotic optimal competitive ratios. Intuitively, these large bins could result in significant wastage of bin capacity, since each bin hosts one item only. In fact, small items may still be placed in large bins to improve the bin utilization. Therefore, we also propose variants of RCP and PPE which eliminate the concept of large bins. Specifically, we modify the algorithms as follows:
\begin{itemize}
    \item All items (irrespective of large or small) go to the general bins first, when the total size of active items in the same category does not exceed the threshold $1/{\sqrt{x}}$ in the general bins.
    \item For each incoming large item (exceeding the threshold), if the category of the item is on, the item goes to the category bins. 
    \item Otherwise, if the item can fit into the remaining capacity of the base bin, it goes to the base bin. After receiving the large item, the base bin is converted to a category bin, since the total size of items in it exceeds $1/2$. 
    \item Finally, if the item cannot fit into the remaining capacity of the base bin, a category bin is opened for the item, and as a result, the category of the item is turned on immediately. 
\end{itemize}

\begin{figure*}
\centering
\begin{minipage}{.47\textwidth}
\centering
\includegraphics[width=\linewidth]{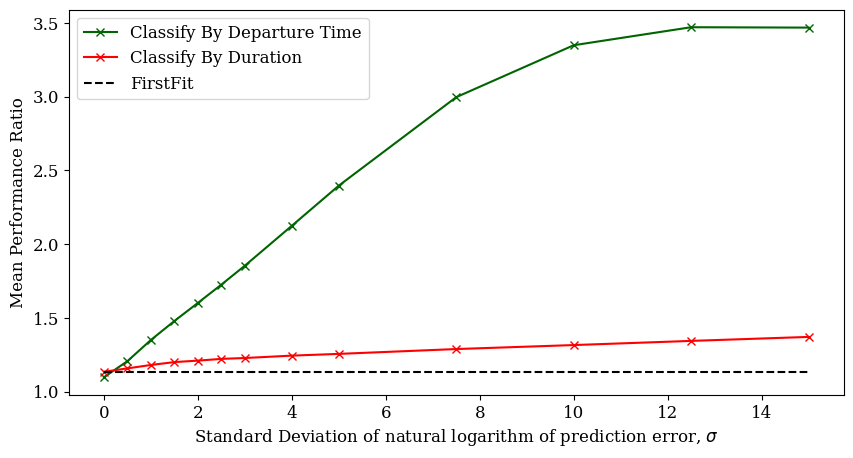}
\caption{\label{fig:errorplot_CB} Classify-By-Departure-Time and Classify-By-Duration} 
\end{minipage}
\quad
\begin{minipage}{.48\textwidth}
\centering
\includegraphics[width=\linewidth]{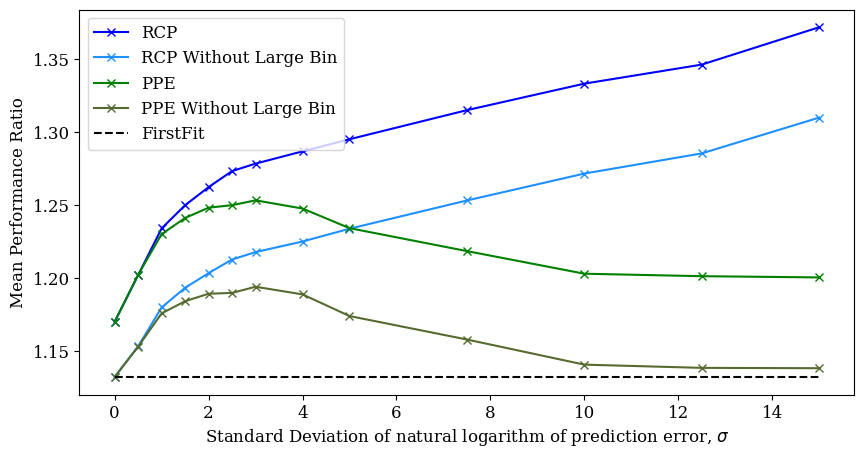}
\caption{\label{fig:errorplotRCP} RCP and PPE} 
\end{minipage}
\end{figure*}

\begin{figure*}
\centering
    \begin{minipage}{.48\textwidth}
    \centering
    \includegraphics[width=\linewidth]{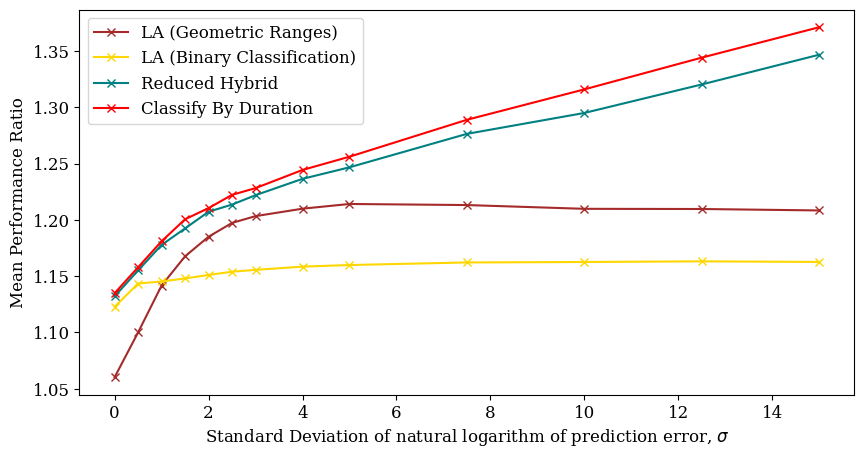}
    \caption{\label{fig:errorplotLA} Lifetime Alignment, Classify-By-Duration and Reduced Hybrid} 
    \end{minipage}
    \quad
    \begin{minipage}{.48\textwidth}
    \centering
        \includegraphics[width=\linewidth]{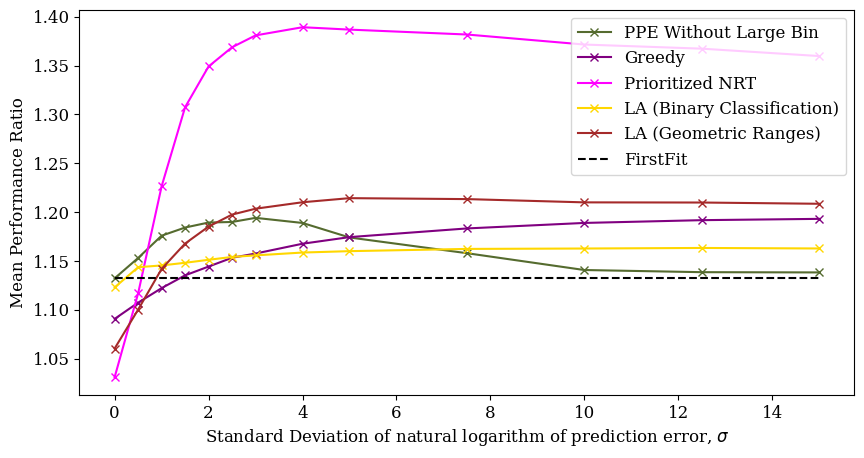}
        \caption{Mean performance ratios of selected algorithms}
        \label{fig:errorplot}   
    \end{minipage}
\end{figure*}

\subsection{Lifetime Alignment (LA)} 

Recently, Barbalho \textit{et al.} \cite{AzureML2023} designed a Lifetime Alignment (LA) algorithm for VM placement with lifetime predictions. It is built upon the Classify-By-Duration strategy, i.e., categorizing items (VMs) according to their predicted durations (lifetimes). The Lifetime Alignment algorithm introduces a new idea to dynamically classify bins (PMs) based on their predicted remaining usage times. The predicted remaining usage time of a bin is given by the duration from the current time to the latest predicted departure time of all items currently placed in the bin (or equivalently the longest predicted remaining duration of all items in the bin). This contrasts the static bin classification used by the Classify-By-Duration, Hybrid, RCP and PPE algorithms, where each bin is always associated with a fixed item category. By the aforesaid dynamic classification, a bin can be associated with different item categories as time goes on. The intuition is to prevent placing an item of long duration into a bin that is about to be closed (all items in it will depart soon, though they had long durations). Lifetime Alignment is an Any Fit algorithm. Let $X_0, X_1, X_2, \dots$ be disjoint ranges of durations for categorization, where $X_0$ contains the smallest values. For each category except $X_0$, an incoming item will first consider the open bins of the same category for placement and if unsuccessful, then consider other open bins. For category $X_0$ (items with shortest durations), an incoming item will consider all open bins together for placement, since these items can help filling the leftover capacity of any bins and do not impact their usage times. Among the open bins considered with enough available capacity, the Best Fit strategy is used for item placement. In our evaluation, we use the $\ell_{\infty}$-norm to compute the fit score (as discussed in Section \ref{sec:bestfit}). We evaluate two different ways of categorization:
\begin{itemize}
    \item \textbf{Binary classification}: There are only two categories: short-lived and long-lived, by setting $X_0 = [0, 120)$ minutes and $X_1 = [120, \infty)$ minutes following \cite{AzureML2023}. 
    \item \textbf{Geometric ranges}: Similar to the RCP and PPE algorithms, we set $X_0 = [0, 1)$ second, and $X_i = [2^{i - 1}, 2^i)$ seconds for each $i \geq 1$. 
\end{itemize}

Neither version of Lifetime Alignment has a bounded competitive ratio. In fact, when all the items are short items belonging to category $X_0$, Lifetime Alignment is reduced to Best Fit which has been proved not to have a bounded competitive ratio \cite{dvp2016ondemand,dbp2014}.

\subsection{Experimental Comparison in Learning-augmented Setting}\label{sec:experimentforlearning}

\textbf{Additional Experiment Setup}: To evaluate learning-augmented algorithms, we further generate a predicted duration of each item (a predicted lifetime of each VM requested) in the experiment setup. Recall that in a \textit{MinUsageTime} DVBP instance, each item $r$ is characterized by an active interval $I(r) = [I(r)^-, I(r)^+)$, where $I(r)^-$ and $I(r)^+$ are the arrival and departure times of the item respectively. The real duration of the item is given by $\mathrm{Rdur}(r) = I(r)^+ - I(r)^-$. We generate a predicted duration of the item by scaling the real duration $\mathrm{Rdur}(r)$ with a multiplicative prediction error. For each item, we randomly generate a multiplicative prediction error $\delta$ according to a log-normal distribution,\footnote{We have also experimented with other distributions of prediction errors and observed similar performance trends. The results using a uniform distribution of prediction errors are provided in Appendix \ref{app:uniform}.} where the mean and standard deviation of $\delta$'s natural logarithm are $\mu$ and $\sigma$ respectively. The log-normal distribution simulates rare but significant prediction failures of the learning algorithm. It was used in the experimentation of other work on learning-augmented algorithms \cite{lykouris2018icml}. On generating $\delta$, the predicted duration is set to $\mathrm{Pdur}(r) = \delta \cdot \mathrm{Rdur}(r)$, which means that the item is predicted to depart at time $I(r)^- + \mathrm{Pdur}(r)$. We set the mean $\mu = 0$ and vary the standard deviation $\sigma$ in the experiment. $\sigma = 0$ means the prediction error is always $e^\mu = e^0 = 1$, i.e., perfect predictions. For each setting, we repeat the experiment 5 times and plot the mean performance ratio.

Besides the learning-augmented algorithms introduced above, we also include key algorithms in the clairvoyant setting in the evaluation. Clairvoyant algorithms can be considered as learning-augmented algorithms that blindly take predictions as truths. That is, they operate based on predicted item durations or departure times directly. Note that the Nearest Remaining Time, Greedy and Lifetime Alignment algorithms exploit the bin's indicated closing time (the latest departure time of all items currently placed in a bin). If all the active items in a bin have their departure times underestimated, it may happen that the bin's indicated closing time based on predictions is earlier than the current time. In this case, we adapt the algorithms to simply take the current time as the bin's indicated closing time. This essentially means that if the prediction underestimates the item duration, after the predicted duration expires, the item is predicted to depart immediately, so its predicted departure time is set to the current time.

\textbf{Experimental Comparison}: For clarity of illustration, we divide algorithms into several groups and plot their performance separately.
Figures \ref{fig:errorplot_CB}, \ref{fig:errorplotRCP}, \ref{fig:errorplotLA}, and \ref{fig:errorplot} compare the mean performance ratios of various algorithms as the standard deviation $\sigma$ increases (prediction errors spread over a wider range and the maximum prediction error increases). We also include the First Fit algorithm (which is the best-performing non-clairvoyant algorithm) in these figures. Since First Fit does not use any predicted information, its performance is a constant and independent of the prediction error.
    
    \textbf{Classify-By-Departure-Time vs. Classify-By-Duration}: Figure \ref{fig:errorplot_CB} shows that the mean performance ratio of Classify-By-Departure-Time increases much faster than Classify-By-Duration as prediction errors spread more widely. This is because the number of item categories created by Classify-By-Departure-Time increases linearly with the maximum prediction error, whereas that created by Classify-By-Duration increases logarithmically. Hence, the former is less robust against prediction errors than the latter, despite that the former outperforms the latter in the clairvoyant setting (perfect predictions) as was shown in Figure \ref{fig:clairBox}. 
    
    \textbf{RCP and PPE}: Figure \ref{fig:errorplotRCP} shows that if predictions are perfect, RCP and PPE behave exactly the same, and if prediction errors are small, they also perform similarly. When prediction errors further increase, PPE starts to outperform RCP, which demonstrates the benefit of adjusting the threshold value according to the prediction error. If prediction errors are very large, the threshold value $\alpha/\sqrt{x}$ set by PPE becomes very high (because $\alpha$ increases linearly with the maximum prediction error whereas $x$ increases logarithmically). As a result, most items go to the general bins, reducing PPE towards First Fit. Thus, PPE's performance ratio even decreases with increasing prediction errors. In contrast, the threshold value $1/\sqrt{x}$ set by RCP drops with increasing prediction errors. Hence, more items go to the category bins. Since there are also more item categories as prediction errors increase, the packing efficiency deteriorates. Thus, RCP's performance ratio continues to increase with prediction errors. Comparing RCP/PPE with their modified versions, it can be seen that the modified versions perform better than the original versions. This is because eliminating large bins improves the bin utilization, thereby reducing the number of bins that need to be opened. Therefore, we shall choose the modified PPE algorithm without large bins for further comparison below. 
        
    \textbf{Lifetime Alignment vs. Classify-By-Duration and Hybrid}: Figure \ref{fig:errorplotLA} compares the two variants of Lifetime Alignment (LA) as well as Classify-By-Duration and Reduced Hybrid. The performance ratio of LA(Binary Classification) increases slower than the other three algorithms as prediction errors increase. This is because LA(Binary Classification) uses the item duration information gently for categorization and hence is less vulnerable to prediction errors. In contrast, the other three algorithms use the item duration information more heavily for categorization and thus are more likely to suffer from prediction errors. Comparing these three algorithms, we can see that LA(Geometric Ranges) consistently outperforms the other two across different prediction errors. This demonstrates that dynamically categorizing the bins can achieve higher packing efficiency than fixing the categories of bins for item placement.
    
    \textbf{Prioritized NRT and Greedy}: 
    Figure \ref{fig:errorplot} shows that Prioritized NRT is not resilient to prediction errors, despite that it has the best performance when predictions are perfect. In contrast, Greedy is more robust than Prioritized NRT against prediction errors. Interestingly, we find a similarity between Greedy and Prioritized NRT -- both algorithms consider bins with indicated closing times after the item's departure time for item placement before other bins, aiming to place the item in a bin without increasing its indicated closing time. The difference is that Greedy lists these bins in decreasing order of indicated closing time, whereas Prioritized NRT lists these bins in ascending order of indicated closing time. In a sense, Greedy is \textit{conservative} in that it prefers bins with indicated closing times far beyond the item's departure time, while Prioritized NRT is \textit{aggressive} since it prefers bins with indicated closing times near to the item's departure time. In the presence of prediction errors, the bins preferred by Greedy are more likely to have real closing times after the item's departure time than the bins preferred by Prioritized NRT. That is why Greedy is more robust against prediction errors than Prioritized NRT.
        
    \textbf{Overall Comparison}: Figure \ref{fig:errorplot} also includes the best-performing algorithms from previous comparisons. We also include First Fit -- the best performer in the non-clairvoyant setting (it is a flat curve because First Fit is oblivious to prediction errors). As seen from Figure \ref{fig:errorplot}, at high prediction errors, the modified PPE and First Fit algorithms perform the best. At medium prediction errors, modified PPE is somewhat inferior to First Fit because modified PPE does not have the Any Fit feature. In this case, the Greedy and LA(Binary Classification) algorithms outperform modified PPE. When predictions are close to perfect, the Prioritized NRT algorithm performs the best. An interesting future direction is to develop an adaptive algorithm that switches between different algorithms based on the prediction errors observed so that it can combine their advantages to optimize performance over a wider range of prediction errors.

%% file: content/3.Conclusion.tex
\section{Conclusion}
\label{sec:conclusion}

We have conducted extensive experiments to evaluate DVBP algorithms for virtual machine placement. We have also proposed and analyzed new algorithms or enhancements including Round Robin Next Fit, Nearest Remaining Time, Reduced Hybrid, and modified RCP/PPE without large bins. Our evaluations show that algorithms conducting fixed classification of bins based on item departure times or durations, while necessary to break the lower bound of Any Fit algorithms in competitive ratios, do not perform as well as Any Fit algorithms empirically. Our proposed Prioritized NRT algorithm not only achieves a near-optimal competitive ratio among Any Fit algorithms, but also demonstrates the best empirical performance in the clairvoyant setting. The modified PPE, Greedy and LA(Binary Classification) algorithms show stronger resilience to prediction errors than other algorithms in the learning-augmented setting. Interesting future directions include (1) constructing adaptive algorithms to take advantage of the strengths while avoiding the drawbacks of different approaches empirically; (2) closing the gap between the lower bound on the competitiveness of DVBP and the competitive ratios achieved by DVBP algorithms in the multi-dimensional case theoretically; (3) analyzing the average-case performance of DVBP algorithms under practical input distributions. 

%% file: content/app-CR.tex
\onecolumn
\appendix

\section{Competitive Analysis}

We introduce some notations to facilitate the presentation of competitive analysis. Given a set of items $\mathcal{R}$, we denote $A(\mathcal{R}, t)$ as the subset of items in $\mathcal{R}$ that are active at a time instant $t$:  
 \begin{equation*}
 A(\mathcal{R}, t) =\{r \mid r \in \mathcal{R} \land t \in I(r)\};
 \end{equation*}
and denote $SS_\infty(\mathcal{R}, t)$ as the sum of the $\ell_\infty$-norms of the size vectors of all items in $A(\mathcal{R}, t)$: $$SS_\infty(\mathcal{R},t) = \sum_{r\in A(\mathcal{R}, t)} \norminf{s(r)}.$$

A set of items $\mathcal{R}'$ is said the $\alpha$-extension of another set of items $\mathcal{R}$, if the items in $\mathcal{R}'$ are created through extending the durations of all items in $\mathcal{R}$ by a multiplicative factor $\alpha$ via changing the departure time of each item, i.e., \begin{align*}
\mathcal{R}' = \{ & r' \mid r \in \mathcal{R} \land s(r') = s(r) \land I^-(r') = I^-(r) 
\land l(I(r'))^+ - l(I(r'))^- = \alpha \cdot (l(I(r))^+ - l(I(r))^-) \}.
\end{align*}

The following propositions are obvious and they are useful in the analysis.

\begin{proposition}
\label{prop:extension}
    If $\mathcal{R}'$ is the $\alpha$-extension of $\mathcal{R}$, it holds that 
    $\int_{-\infty}^{\infty}{SS_\infty(\mathcal{R}', t)} \, \mathrm{d}t = \alpha \cdot  \int_{-\infty}^{\infty}{SS_\infty(\mathcal{R}, t) \, \mathrm{d}t}$.
\end{proposition}

\begin{proposition}[Relationship between $\ell_\infty$-norms] For any set of $d$-dimensional vectors $\mathcal{V}$, it holds that
$$\frac{1}{d} \sum_{\vec{v}\in\mathcal{V}}{\norminf{\vec{v}}}  \leq \norminf{\sum_{\vec{v}\in\mathcal{V}}{\vec{v}}} \leq  \sum_{\vec{v}\in\mathcal{V}}{\norminf{\vec{v}}}. $$
\label{prop:inftynorm}
\end{proposition}

Denote $\sspan(\mathcal{R})$ as the {\em time span} of all items, i.e., the duration in which at least one item is active.
Proposition \ref{prop:lowerbound} follows from the lower bound \eqref{eq:lowerbound} in Section \ref{sec:problem} and Proposition \ref{prop:inftynorm} directly.

\begin{proposition}[Lower bound of $\OPT$]
\label{prop:lowerbound}
The total bin usage time $\OPT$ of an optimal packing solution for a set of items $\mathcal{R}$ has the
following lower bounds: 

$$\OPT \geq \int_{-\infty}^{\infty}{\ceil{\norminf{\sum_{r \in A(\mathcal{R}, t)}{s(r)}}}}\, \mathrm{d}t \geq \frac{1}{d} \cdot \int_{-\infty}^{\infty}SS_\infty(\mathcal{R},t) \, dt,$$
$$\OPT \geq \int_{-\infty}^{\infty}{\ceil{\norminf{\sum_{r \in A(\mathcal{R}, t)}{s(r)}}}}\, \mathrm{d}t \geq \int_{-\infty}^{\infty} \mathbbm{1}_{A(\mathcal{R}, t) \neq \emptyset} \, \mathrm{d}t = \sspan(\mathcal{R}).$$
\end{proposition}

Without loss of generality, we assume that the minimum item duration is $1$ and the maximum item duration is $\mu$.

\input{content/CR/RRNextFit.tex}
\input{content/CR/greedy.tex}
\input{content/CR/HA.tex}

%% file: content/CR/RRNextFit.tex
\subsection{Round Robin Next Fit}
\label{sec:rrnextfitanalysis}

\begin{theorem}
    The Round Robin Next Fit algorithm has a competitive ratio of $(2\mu+1) d+1$.
\end{theorem}

\begin{proof}
    In Round Robin Next Fit, a bin can receive new items only when it is pointed by the flag indicating the open bin currently under consideration. Thus, for each bin, we can partition its usage time into receiving periods and non-receiving periods when it is pointed and not pointed by the flag respectively. That is, a receiving period of a bin starts when the bin receives an arriving item and the previous arriving item was placed in another bin (or the arriving item is the first input item), and it ends when a different bin starts to receive items. Hence, each bin has alternating receiving and non-receiving periods in its usage time. 

    Note that at most one bin is in its receiving period at any time. Thus, the total length of the receiving periods of all bins is bounded by $\OPT$, since at least one bin must be used by an optimal packing solution whenever at least one item is active.
    
    Next, we bound the number of bins in their non-receiving periods at each time instant $t$, and take the integral over the time horizon to bound the total length of the non-receiving periods of all bins. At a time instant $t$, if there are a total of $k$ bins in use, exactly one bin is in its receiving period and all other $k-1$ bins are in their non-receiving periods. Denote $b_1, b_2,\dots,b_{k-1}$ as the bins that are in their non-receiving periods. Let $t_i$ be the latest time instant before $t$ when bin $b_i$ switched from its receiving period to its non-receiving period. Obviously, $t_i > t - \mu$ holds, because if bin $b_i$ did not receive any item during the interval $(t - \mu, t]$, it would not remain open at time $t$. By the algorithm definition, at time $t_i$, there must be an arriving item $r_i$ that cannot fit into bin $b_i$ and is hence placed in a different bin, causing $b_i$ to enter its non-receiving period. Thus, denoting $\mathcal{R}_{i}$ as the set of all items placed in bin $b_i$, we have the following inequality:
    \begin{eqnarray*}
        1 < \norminf{s(r_i) + \sum_{r \in A(\mathcal{R}_{i},t_i)}{s(r)} } 
        & \leq & \norminf{s(r_i)} + \sum_{r \in A(\mathcal{R}_{i},t_i)}{\norminf{s(r)} }  
        = \norminf{s(r_i)}+ SS_\infty(\mathcal{R}_{i},t_i).
    \end{eqnarray*}
    
   Adding up the above inequality over all $i = 1,2,\dots,k - 1$ would bound $k-1$ (the number of bins in their non-receiving periods). Let $\mathcal{R}$ be the set of all the input items. To bound the sum of $\norminf{s(r_i)}$'s, we define $\mathcal{R}'$ as the $\mu$-extension of $\mathcal{R}$. It can be shown that the $\mu$-extension of each item $r_i$ in $\mathcal{R}'$ is active at time $t$. Note that item $r_i$ arrives at time $t_i$ and will depart at time $t_i + l_i$, where $l_i$ is the duration of $r_i$. Hence, if the duration of $r_i$ is extended by a factor of $\mu$, $r_i$ will depart at time $t_i + l_i \cdot \mu > t$, because $l_i$ is at least $1$ and $t_i > t - \mu$. Therefore, we have
   $$\sum_{i=1}^{k-1}{\norminf{s(r_i)}} \leq SS_\infty(\mathcal{R}',t).$$
    
   To bound the sum of $SS_\infty(\mathcal{R}_i,t_i)$'s, we define $\mathcal{R}''$ as the $(\mu + 1)$-extension of $\mathcal{R}$. It can be shown that for each item in $\mathcal{R}$ active at time $t_i$, its $(\mu + 1)$-extension in $\mathcal{R}''$ must be active at time $t$. In fact, let $l$ be the duration of an item $r \in \mathcal{R}$. Since $r$ is active at $t_i$, it must arrive later than $t_i - l$. Hence, if the duration of $r$ is extended by a factor of $\mu+1$, $r$ would depart later than $t_i - l + (\mu+1) \cdot l = t_i + \mu \cdot l > t - \mu + \mu \cdot l \geq t$, where the two inequalities are due to $t_i > t-\mu$ and $l \geq 1$ respectively. Hence, the $(\mu+1)$-extension of $r$ must be active at time $t$. Thus,
   $$\sum_{i=1}^{k-1}{SS_\infty(\mathcal{R}_i,t_i)} \leq SS_\infty(\mathcal{R}'',t).$$
    
   In summary, we obtain
   \begin{equation*}
   k-1 < \sum_{i=1}^{k-1}{\norminf{s(r_j)}} + \sum_{i=1}^{k-1}{SS_\infty(\mathcal{R}_i,t_i)}
   \leq SS_\infty(\mathcal{R}',t) + SS_\infty(\mathcal{R}'',t).
   \end{equation*}

   Taking the integral over the time horizon, the total length of the non-receiving periods of all bins is bounded by
   \begin{eqnarray*}
       \lefteqn{\int_{-\infty}^{\infty}  \left( SS_\infty(\mathcal{R}',t) + SS_\infty(\mathcal{R}'',t) \right)
       \, \mathrm{d}t}\\
       & = & \mu \cdot \int_{-\infty}^{\infty}{SS_\infty(\mathcal{R},t)} \,\mathrm{d}t  + (\mu + 1) \cdot \int_{-\infty}^{\infty}{SS_\infty(\mathcal{R},t)} \,\mathrm{d}t \qquad \text{(by Proposition \ref{prop:extension})} \\
       & \leq & \mu d \cdot \OPT + (\mu +1) d \cdot \OPT \qquad \text{(by Proposition \ref{prop:lowerbound})} \\
       & = & (2\mu+1) d  \cdot \OPT.
    \end{eqnarray*}
Together with the receiving periods, the total bin usage time is bounded by $((2\mu+1) d +1)  \cdot \OPT$.
\end{proof}

\begin{theorem}
    The competitive ratio of Round Robin Next Fit algorithm is at least $2\mu d$. 
\end{theorem}
\begin{proof}
We construct an adversary to show that the Round Robin Next Fit algorithm produces a total bin usage time of $2\mu d \cdot \OPT$. The adversary releases items in $r+1$ rounds. Let $\epsilon$ be a sufficiently small number and $k$ be a sufficiently large even number satisfying $dk\epsilon \leq 1$. Denote $\tau$ as a sufficiently small time unit. 

Initially, the adversary releases $2dk$ items at time $0$, and the objective is to make the algorithm open $dk$ bins. The adversary releases these items sequentially. All the odd-indexed items have size $1-\epsilon$ in all dimensions and a duration of $1$. All the even-indexed items have size $\epsilon$ in all dimensions and a duration of $1 + \tau$. It can be seen that one bin will be opened by the Round Robin Next Fit algorithm to accept one even-indexed item and one odd-indexed item (as the sum of their sizes is equal to 1 in all dimensions). Hence, the Round Robin Next Fit algorithm will open $dk$ bins and so does the optimal algorithm.

Next, the adversary releases $r$ rounds of items, where round $i$ is released just before time $1+ \tau + (i - 1)\mu$. We will show that at the beginning of each round, there are $dk$ bins in use, each with an item of size $\epsilon$ in all dimensions. It is clearly true for round $1$. In each round $i$, $2dk$ items are released. The items are partitioned into $d$ groups where the items indexed from $2(j-1)k + 1$ to $2jk$ belongs to the $j$-th group. The $2k$ items in every group are further divided into $\frac{k}{2}$ segments, each including four items. The four items in every segment of the $j$-th group are designed as follows: The first item has size $d\epsilon$ in dimensions $j$ and $j - 1$ (if $j=1$ then in dimensions $1$ and $d$) and $0$ in all other dimensions and has a duration of $1$. The second and the fourth items have size $\epsilon$ in all dimensions and a duration of $\mu$. Conceptually, these even-indexed items are small in size and long in duration. The third item has size $1-d\epsilon$ in dimension $j$ and $\epsilon$ in all other dimensions and has a duration of $1$. It can be shown that the Round Robin Next Fit algorithm will put one even-indexed item in each of the $dk$ bins and thus causing them to continue to be in use until the end of the round (which is carried over to the next round if there is one). For every segment of four items, the first item and second item will be placed into the bin which the flag is pointing to. The third item cannot be placed in the same bin because doing so would cause dimension $j$ to have a total size of $\epsilon + d\epsilon + \epsilon + (1 - d\epsilon) = 1 + 2\epsilon$ which is greater than 1. Hence, the third item will be placed into the next bin together with the fourth item. At the beginning of the $j$-th group of items where $j > 1$, since the first item has size $d\epsilon$ in dimension $j-1$, this causes the first item to be placed into the next bin for the same reason. Hence, the even-indexed items are all placed in distinct bins. In summary, the $dk$ bins are in use throughout all the rounds. Thus, Round Robin Next Fit has a total bin usage time of $dk(1 + \tau + r\mu)$.

In the optimal algorithm, all the even-indexed items can be placed into the same bin because $dk\epsilon \leq 1$. This bin will be in use for $\mu$ time. The first items in all segments of all groups can be placed into the same bin because in each dimension, there are two groups each with $\frac{k}{2}$ items having non-zero size in the dimension. The total size of these items in the dimension is $\frac{k}{2} \cdot 2 \cdot (d\epsilon) = kd\epsilon \leq 1$. Finally, the fourth items in the same segments of all groups can be placed into one bin because their total size in each dimension is $1 - d\epsilon + (d-1)\epsilon = 1 - \epsilon < 1$. Since there are $\frac{k}{2}$ segments in each group, $\frac{k}{2}$ bins are needed. Hence, in each round, one bin and $\frac{k}{2}$ bins are used for a duration of $1$ and one bin is used for a duration of $\mu$. This is repeated for $r$ rounds. Thus, the optimal total bin usage time is at most $dk + r(1 + \frac{k}{2} + \mu)$, where $dk$ is the bin usage time before round $1$.

Therefore, as $k$ and $r$ grow to infinity, the ratio of the total bin usage time by the Round Robin Next Fit algorithm to that by the optimal algorithm is
\begin{eqnarray*}
    \lim_{k,r \rightarrow \infty}\frac{dk(1 + r\mu + \tau)}{dk + r(1 + \frac{k}{2} + \mu)} & = & \lim_{k,r \rightarrow \infty }\frac{\frac{d}{r} + d\mu + \frac{d\tau}{r}}{\frac{d}{r} + \frac{1}{k} + \frac{1}{2} + \frac{\mu}{k}} 
    = \frac{d\mu}{\frac{1}{2}} = 2d\mu. 
\end{eqnarray*}
\end{proof}

%% file: content/CR/greedy.tex
\subsection{Greedy and Prioritized NRT}
\label{sec:greedyandnrtanalysis}

\begin{theorem} \label{theorem:PrioritizedNRT}
    The Greedy and Prioritized NRT algorithms have a competitive ratio of $(\mu + 2)d + 1$.
\end{theorem}

\begin{proof}

    We bound the number of bins used by the algorithms at each time instant $t$, and integrate it over the time horizon. 

    At a time instant $t$, suppose a total of $k$ bins $b_1,b_2,\dots,b_k$ are in use. Each of these bins must have received at least one item with arrival time in the interval $(t-\mu,t]$ and departure time later than $t$. Let $r_i$ denote the first item placed into bin $b_i$ that arrived in $(t-\mu,t]$ and will depart later than $t$. That is, $r_i$ is the first item that extends the indicated closing time of $b_i$ to be greater than $t$. Let $t_i$ denote the arrival time of $r_i$. Note that $t_i$ is within $(t-\mu,t]$, and $r_i$ may be the item that opened bin $b_i$ (i.e., the very first item placed into $b_i$).

    Without loss of generality, assume that $b_1,b_2,\dots,b_k$ are arranged in ascending order of $t_i$. For each $i \geq 2$, at time $t_i$, the indicated closing time of bin $b_{i-1}$ is greater than $t$, and the indicated closing time of bin $b_{i}$ is at most $t$ by definition. Therefore, according to the Greedy algorithm, bin $b_{i-1}$ has higher priority for accepting item $r_i$ (when it arrives at $t_i$) than bin $b_i$, i.e., the algorithm must have tried placing $r_i$ into $b_{i-1}$ but could not and thus placed $r_i$ into $b_i$. Note that $b_{i-1}$ and $b_{i}$ are not necessarily consecutive in the ordering of bins for placing $r_i$.
    Similarly, by the Prioritized NRT algorithm, when placing item $r_i$  with departure time greater than $t$, it tries all bins with indicated closing times greater than $t$ before those with indicated closing times at most $t$. Thus, bin $b_{i-1}$ also has higher priority for accepting item $r_i$ than bin $b_i$.

    Denote $\mathcal{R}_{i}$ as the set of all items placed in bin $b_i$. Since item $r_i$ cannot fit into bin $b_{i-1}$, we have 
    \begin{eqnarray*}
        1 < \norminf{ s(r_i) +  \sum_{r\in A(\mathcal{R}_{i-1}, t_i)}s(r) } 
        & \leq & \norminf{s(r_i)} + \sum_{r\in A(\mathcal{R}_{i-1}, t_i)}\norminf{s(r)}
        = \norminf{s(r_i)} + SS_\infty(\mathcal{R}_{i-1},t_i).
    \end{eqnarray*}
     
    Adding up the above inequality over all $i \geq 2$ would bound the number of bins $k$. Let $\mathcal{R}$ be the set of all the input items. For the sum of $\norminf{s(r_i)}$'s, since each item $r_i$ is active at time $t$, we have 
    $$ \sum_{i=2}^k{\norminf{s(r_i)}} \leq SS_\infty(\mathcal{R},t).$$

    Next, we bound the sum of $SS_\infty(\mathcal{R}_{i-1},t_i)$'s. Let $\mathcal{R}'$ be the $(\mu + 1)$-extension of $\mathcal{R}$. It can be shown that for each item in $\mathcal{R}$ active at time $t_i$, its $(\mu + 1)$-extension in $\mathcal{R}'$ must be active at time $t$. In fact, let $l$ be the duration of an item $r \in \mathcal{R}$. Since $r$ is active at $t_i$, it must arrive later than $t_i - l$. Hence, if the duration of $r$ is extended by a factor of $\mu+1$, $r$ would depart later than $t_i - l + (\mu+1) \cdot l = t_i + \mu \cdot l > t - \mu + \mu \cdot l \geq t$, where the two inequalities are due to $t_i > t - \mu$ and $l \geq 1$ respectively. Hence, the $(\mu+1)$-extension of $r$ must be active at time $t$. Thus,
 $$ \sum_{i=2}^k{SS_\infty(\mathcal{R}_{i-1},t_i)} \leq SS_\infty(\mathcal{R}',t).$$

    In summary, the number of bins $k$ can be bounded as follows:
    \begin{equation*}
        k  \leq  \sum_{i = 2}^k{\norminf{s(r_i)}} + \sum_{i=2}^k{SS_\infty(\mathcal{R}_{i-1},t_i)} + 1  
        \leq  SS_\infty(\mathcal{R},t) + SS_\infty(\mathcal{R}',t) + 1.
    \end{equation*}

    Finally, we take the integral over the time horizon. Note that the integral of the constant $1$ gives the time span $\sspan(\mathcal{R})$ which is bounded by $\OPT$ according to Proposition \ref{prop:lowerbound}. Thus, the total bin usage time is bounded by
    \begin{eqnarray*}
        \lefteqn{\int_{-\infty}^{\infty}{  \left( SS_\infty(\mathcal{R},t) + SS_\infty(\mathcal{R}',t) 
        \right)\, \mathrm{d}t } + \OPT} \\
         & = & \int_{-\infty}^{\infty}{  SS_\infty(\mathcal{R},t)} \, \mathrm{d}t + (\mu+1) \cdot \int_{-\infty}^{\infty}{SS_\infty(\mathcal{R},t)} \, \mathrm{d}t + \OPT \qquad \text{(by Proposition \ref{prop:extension})} \\ 
         & \leq & ((\mu + 2) d + 1) \cdot \OPT. \qquad \text{(by Proposition \ref{prop:lowerbound})}       
    \end{eqnarray*}

\end{proof}

%% file: content/CR/HA.tex
\subsection{Hybrid and Reduced Hybrid}
\label{sec:reducedhybridanalysis}

\begin{theorem}
    The Hybrid and Reduced Hybrid algorithms achieve a competitive ratio of $O(d\sqrt{\log\mu})$.
\end{theorem}

\begin{proof}
We first prove the result for the Reduced Hybrid algorithm. We study the usage times of general bins and category bins separately. 

First, we bound the number of general bins used at any time instant. 
Suppose there are $k$ general bins in use at a given time instant. Let $g_1, g_2, \dots, g_k$ denote these $k$ general bins in ascending order of opening time.
Denote $\mathcal{G}_{i}$ as the set of all items placed in bin $g_i$.
Let $r_k$ denote the first item placed into bin $g_k$ (i.e., $r_k$ is the item that opened bin $g_k$).
Note that there are at most $\ceil{\log\mu}+1$ categories. Right after $g_k$ is placed at its arrival time $I(r_k)^-$, by the algorithm definition, the $\ell_\infty$-norm of the aggregate size of all the active items (including $r_k$) in the general bins must be bounded by \begin{eqnarray*}
\frac{1}{2\sqrt{1}} + \frac{1}{2\sqrt{2}} + \cdots + \frac{1}{2(\sqrt{\ceil{\log\mu}+1})} & \leq & \frac{1}{2} + \int^{\ceil{\log\mu}+1}_{x=1} \frac{1}{2\sqrt{x}}\,\textrm{d}t \\ & = & \frac{1}{2} + (\sqrt{\ceil{\log\mu}+1} - 1) \\ & = & \sqrt{\ceil{\log\mu}+1} - \frac{1}{2}.
\end{eqnarray*}
It follows from Proposition \ref{prop:inftynorm} that 
\begin{equation}
\sum_{j=1}^k SS_\infty(\mathcal{G}_j,I(r_k)^-) \leq d \norminf{\sum_{j=1}^k \sum_{r \in A(\mathcal{G}_j,I(r_k)^-)}{s(r)}} \leq d \left(\sqrt{\ceil{\log\mu}+1} - \frac{1}{2}\right).
\label{eq:3}
\end{equation}
Moreover, we have $\norminf{s(r_k)} \leq \frac{1}{2}$, since $r_k$ is placed into the general bins and $\frac{1}{2}$ is the largest threshold value among all categories. 
By the First Fit rule, for each $j=1,2,\dots,k-1$, since item $r_k$ cannot fit into bin $g_{j}$, we have \begin{equation*}
1 < \norminf{\sum_{r \in A(\mathcal{G}_j,I(r_k)^-)}{s(r)} + s(r_k)} \leq SS_\infty(\mathcal{G}_j,I(r_k)^-) + \norminf{s(r_k)} \leq SS_\infty(\mathcal{G}_j,I(r_k)^-) + \frac{1}{2}, 
\end{equation*}
which implies $SS_\infty(\mathcal{G}_j,I(r_k)^-) > \frac{1}{2}$. 
Therefore, by \eqref{eq:3}, we have 
\begin{eqnarray*}
d \left(\sqrt{\ceil{\log\mu}+1} - \frac{1}{2}\right) & \geq & \sum_{j=1}^k SS_\infty(\mathcal{G}_j,I(r_k)^-) \\ & \geq & (k-2)\cdot \frac{1}{2} + SS_\infty(\mathcal{G}_{k-1},I(r_k)^-) + \norminf{s(r_k)} \\ & > & (k-2)\cdot \frac{1}{2} + 1.
\end{eqnarray*}
It follows that $k < 2d\sqrt{\ceil{\log\mu}+1} - d$. Hence, the total usage time of general bins is bounded by $(2d\sqrt{\ceil{\log\mu}+1} - d) \cdot \OPT$, since at least one bin must be used by an optimal packing solution whenever at least one item is active.

    Next, we bound the total usage time of category bins. We partition all the input items $\mathcal{R}$ into those placed in general bins $\mathcal{R}_{GN}$ and those placed in category bins $\mathcal{R}_{CT}$. $\mathcal{R}_{GN}$ (resp. $\mathcal{R}_{CT}$) is further partitioned into $\mathcal{R}_{GN,i}$'s (resp. $\mathcal{R}_{CT, i}$'s) including items of respective category $i$'s. 
    
    Consider any category $i$. Suppose there are $c_i$ category-$i$ bins $b_1, b_2, \ldots, b_{c_i}$ in use at a time instant $t$. 
    Since the items of category $i$ have durations at most $2^i$, each bin $b_j$ must have received at least one item with arrival time in the interval $(t-2^i,t]$ and departure time later than $t$. Let $r_j$ denote the first item placed into bin $b_j$ that arrived in $(t-2^i,t]$ and will depart later than $t$. Let $t_j$ denote the arrival time of $r_j$. 
    Without loss of generality, assume that $b_1,b_2,\dots,b_{c_i}$ are arranged in ascending order of $t_j$.    
    By the First Fit rule, for each $j = 2,3,\dots,c_i$, the fact that item $r_j$ is placed into bin $b_j$ implies that $r_j$ cannot fit into bin $b_{j-1}$. Denoting $\mathcal{R}_{j}$ as the set of all items placed in bin $b_j$, we have
    \begin{equation}
        1 < \norminf{s(r_j) + \sum_{r \in A(\mathcal{R}_{j-1},t_j)}{s(r)}} 
        \leq \norminf{s(r_j)} + SS_\infty(\mathcal{R}_{j-1},t_j).
        \label{eq:1}
    \end{equation}
    As for bin $b_1$, the fact that item $r_1$ is placed into bin $b_1$ and not general bins implies that
    \begin{equation*}
    \frac{1}{2\sqrt{\ceil{\log\mu}+1}} \leq \frac{1}{2\sqrt{i}} <  \norminf{s(r_1) + \sum_{r \in A(\mathcal{R}_{GN,i},t_1)}{s(r)} }  
    <  \norminf{s(r_1)} + SS_\infty(\mathcal{R}_{GN,i},t_1). 
    \end{equation*}   
    Hence,
    \begin{equation}
       1 <  2\sqrt{\ceil{\log\mu}+1} \cdot \norminf{s(r_1)} + 2\sqrt{\ceil{\log\mu}+1} \cdot SS_\infty(\mathcal{R}_{GN,i},t_1).
       \label{eq:2}
    \end{equation}
    By adding up the inequalities \eqref{eq:1} and \eqref{eq:2}, we can bound $c_i$, the number of category-$i$ bins in use at time $t$. First, we bound the sum of $\norminf{s(r_j)}$'s. Since each item $r_j$ is still active at time $t$, 
    we have 
    $$2\sqrt{\ceil{\log\mu}+1} \cdot \norminf{s(r_1)} + \sum_{j=2}^{c_i}{\norminf{s(r_j)}} \leq 2\sqrt{\ceil{\log\mu}+1} \cdot SS_\infty(\mathcal{R}_{CT,i},t).$$
    Next, we bound the sum of $SS_\infty(\mathcal{R}_{j-1},t_j)$'s. Let $\mathcal{R}'_{CT,i}$ (resp. $\mathcal{R}'_{GN,i}$) be the $3$-extension of $\mathcal{R}_{CT,i}$ (resp. $\mathcal{R}_{GN,i})$. 
    It can be shown that for each item in $\mathcal{R}_{CT,i}$ or $\mathcal{R}_{GN,i}$ active at time $t_j$, its 3-extension in $\mathcal{R}'_{CT,i}$ or $\mathcal{R}'_{GN,i}$ is active at time $t$. In fact, let $l$ be the duration of an item $r \in \mathcal{R}_{CT,i} \cup \mathcal{R}_{GN,i}$. Since $r$ is a category-$i$ item, we have $l \in [2^{i-1},2^i)$. Since $r$ is active at $t_j$, it must arrive after $t_j - l$ and thus its $3$-extension must depart after $t_j - l + 3l = t_j + 2l \geq t_j + 2 \cdot 2^{i - 1} = t_j + 2^{i} \geq t$, where the last inequality is because $t_j \in (t-2^i,t]$. 
    Thus, we have the following inequalities:
    $$\sum_{j=2}^{c_i} SS_\infty(\mathcal{R}_{j-1},t_j) \leq SS_\infty(\mathcal{R}'_{CT,i}, t),$$
    $$SS_\infty(\mathcal{R}_{GN,i},t_1) \leq SS_\infty(\mathcal{R}'_{GN,i},t).$$
    In summary, we obtain
    \begin{eqnarray*}
        c_i & < & 2\sqrt{\ceil{\log\mu}+1} \cdot \norminf{s(r_1)} + \sum_{j=2}^{c_i}{\norminf{s(r_j)}} 
        + 2\sqrt{\ceil{\log\mu}+1} \cdot SS_\infty(\mathcal{R}_{GN,i},t_1) + \sum_{j=2}^{c_i} SS_\infty(\mathcal{R}_{j-1},t_j) \\
        & \leq & 2\sqrt{\ceil{\log\mu}+1} \cdot {SS_\infty(\mathcal{R}_{CT,i},t)} + 2\sqrt{\ceil{\log\mu}+1} \cdot SS_\infty(\mathcal{R}'_{GN,i},t) +{SS_\infty(\mathcal{R}'_{CT,i},t)}.
    \end{eqnarray*}

Finally, integrating over the time horizon, the total usage time of category bins is bounded by 
    \begin{eqnarray*}
        \lefteqn{\int_{-\infty}^{\infty}{\sum_{i=1}^{\ceil{\log\mu}+1}\left(2\sqrt{\ceil{\log\mu}+1} \cdot {SS_\infty(\mathcal{R}_{CT,i},t)} + 2\sqrt{\ceil{\log\mu}+1} \cdot SS_\infty(\mathcal{R}'_{GN,i},t) +{SS_\infty(\mathcal{R}'_{CT,i},t)}\right)} \, \mathrm{d}t} \\
        & = & \int_{-\infty}^{\infty} \left(2\sqrt{\ceil{\log\mu}+1} \cdot SS_\infty(\mathcal{R}_{CT},t) + 2\sqrt{\ceil{\log\mu}+1} \cdot SS_\infty(\mathcal{R}'_{GN},t) + SS_\infty(\mathcal{R}'_{CT},t)\right) \, \mathrm{d}t\\
& =& \int_{-\infty}^{\infty} \left(2\sqrt{\ceil{\log\mu}+1} \cdot SS_\infty(\mathcal{R}_{CT},t) + 6\sqrt{\ceil{\log\mu}+1} \cdot SS_\infty(\mathcal{R}_{GN},t) + 3\cdot SS_\infty(\mathcal{R}_{CT},t)\right) \, \mathrm{d}t\\
& \leq & \int_{-\infty}^{\infty} 6\sqrt{\ceil{\log\mu}+1} \cdot SS_\infty(\mathcal{R},t) \, \mathrm{d}t\\
& \leq & 6d\sqrt{\ceil{\log\mu}+1} \cdot \OPT. \qquad \text{(by Proposition \ref{prop:lowerbound})}  
    \end{eqnarray*}

Overall, the total usage time of all bins is bounded by $(8d\sqrt{\ceil{\log\mu}+1}-d) \cdot \OPT$.

Proving the result for the Hybrid algorithm is almost the same as the Reduced Hybrid algorithm. Recall that in the Hybrid algorithm, each category $(i,c)$ includes all the items with durations in the range $[2^{i-1}, 2^i)$ for an integer $i$ and arrival times in the interval $[(c - 1) \cdot 2^i, c \cdot 2^i)$ for an integer $c$. For any given $i$, since the duration of each item is at most twice the length of the interval for arrival times, there can be at most two types of items (with the same $i$) active at the same time. We can apply the same analysis for each category $(i,c)$ in Hybrid as that for each category $i$ in Reduced Hybrid. The result of the total bin usage time will just have an additional multiplicative factor of $2 = O(1)$. Thus, it will not change the asymptotic order of the competitive ratio.
\end{proof}

%% file: content/app-huawei.tex
\subsection{Experimentation with Huawei Dataset} \label{app:huawei}

As mentioned in Section \ref{sec:setup}, we have also experimented with the Huawei-East-1 dataset \cite{sheng2021vmagent} (\url{https://github.com/huaweicloud/VM-placement-dataset}). The dataset has a schema including the following fields:

\begin{itemize}
    \item vmid (the id of the VM created or terminated)
    \item time (the timestamp of the event)
    \item CPU core, memory (the resources used by the VM)
    \item type (event type -- creation or termination)
\end{itemize}

We clean the dataset by keeping only the VMs with both creation and termination events. There are a total of 116,313 such VMs for our experimentation. The number of resource dimensions is $d=2$ (CPU core and memory). The dataset does not provide any information about the resource capacities of physical machines (PM). Since the largest VM in the dataset uses 64 CPU cores and 128 GB memory, we set the CPU capacity to 64, 100 or 128 cores and the memory capacity to 128, 200 or 256 GB (combining them gives 9 different PM capacities). 

We run different packing algorithms for each of the assumed PM capacities and present the box plot of the performance ratios obtained in Figure \ref{fig:huawei}.
The evaluation results using the Huawei Cloud dataset show largely similar trends to those using the Microsoft Azure dataset. 
One noticeable difference is that for the Huawei dataset, $\ell_2$-norm generally works the best for Best Fit. Hence, $\ell_2$-norm is used to represent Best Fit in Figure \ref{fig:nclairhw}.

\begin{figure}[!t]
    \centering

    \begin{subfigure}{0.32\textwidth}
        \centering
        \includegraphics[width = \textwidth]{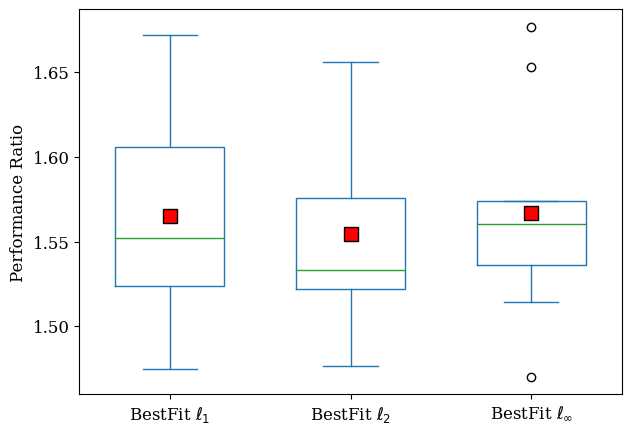}
        \caption{Performance of Best Fit}
        \label{fig:bestBoxhw}
    \end{subfigure}
    ~
    \begin{subfigure}{0.5\textwidth}
        \centering
        \includegraphics[width = \textwidth]{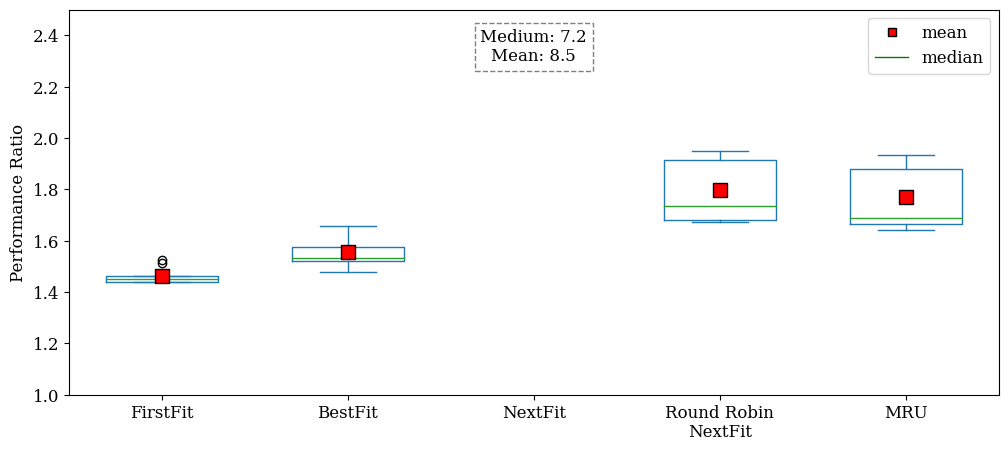}
        \caption{Performance of non-clairvoyant algorithms}
        \label{fig:nclairhw}
    \end{subfigure}

    \begin{subfigure}{0.35\textwidth}
        \centering
        \includegraphics[width = \textwidth]{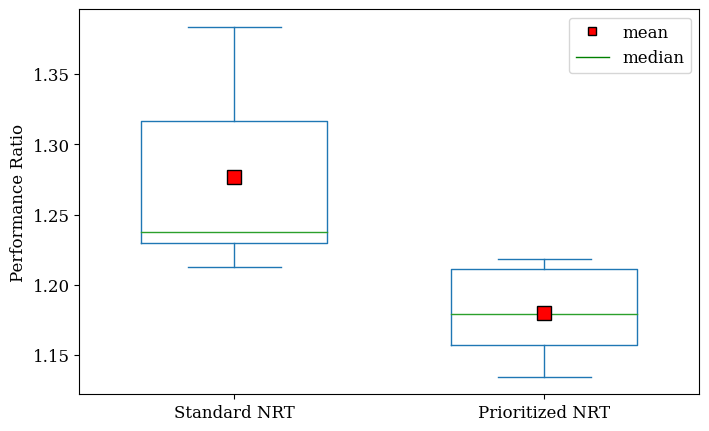}
        \caption{Performance of NRT}
        \label{fig:NRThw}
    \end{subfigure}
    ~
    \begin{subfigure}{0.4\textwidth}
        \centering
        \includegraphics[width = \textwidth]{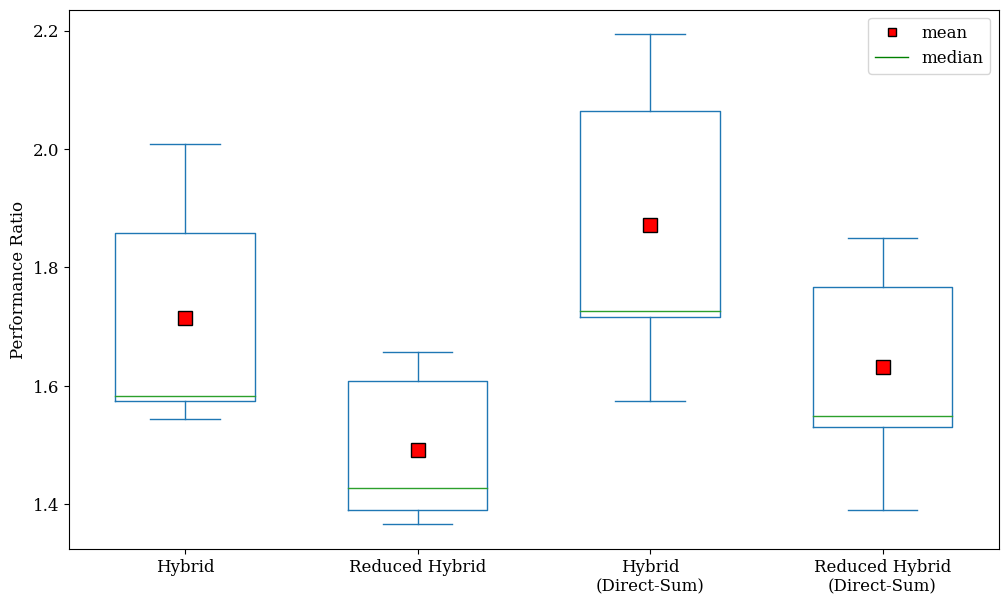}
        \caption{Performance of Hybrid algorithms}
        \label{fig:HAhw}
    \end{subfigure}

    \begin{subfigure}{0.4\textwidth}
        \centering
        \includegraphics[width=\textwidth]{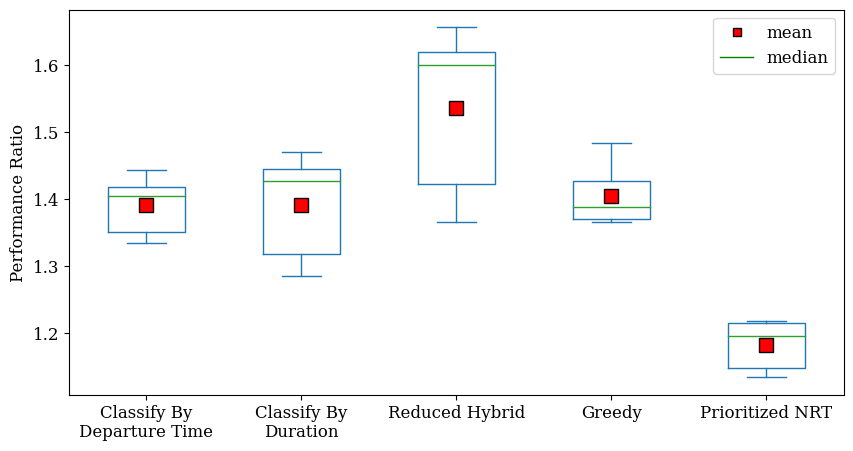}
        \caption{Performance of clairvoyant algorithms}
        \label{fig:clairhw}
    \end{subfigure}

    \begin{subfigure}{.40\textwidth}
        \centering
        \includegraphics[width=\linewidth]{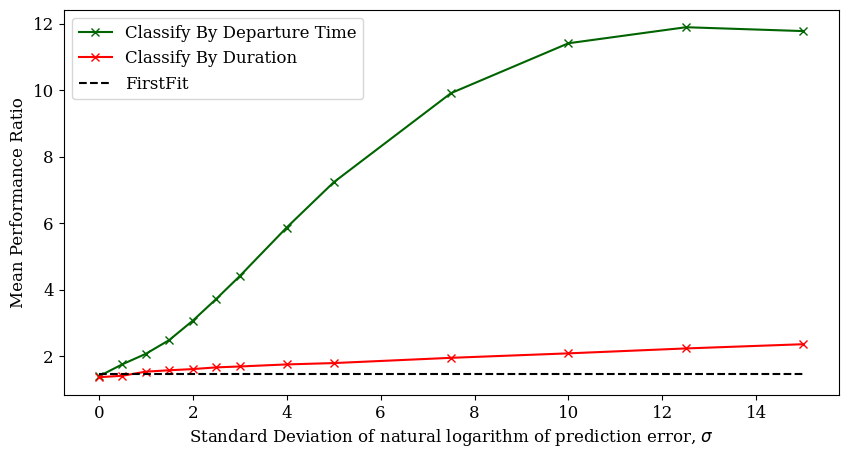}
        \caption{Classify-By-Departure-Time and Classify-By-Duration}
        \label{fig:errorplot_CB_huawei}
    \end{subfigure}
    ~
    \begin{subfigure}{.4\textwidth}
        \centering
        \includegraphics[width=\linewidth]{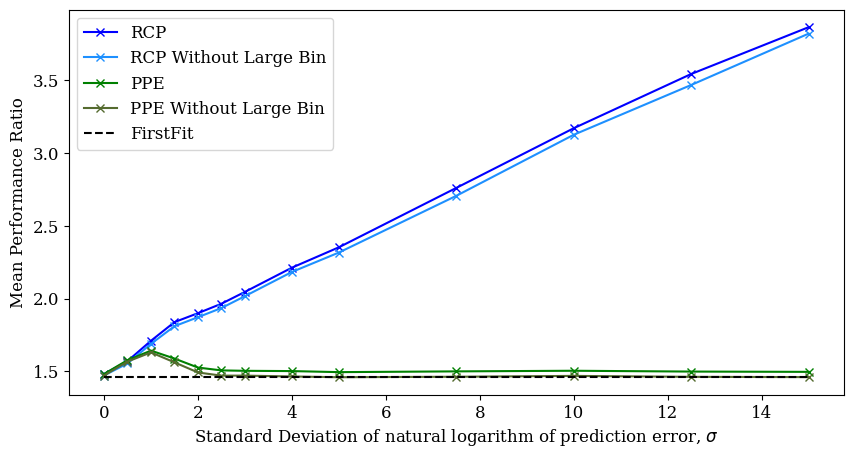}
        \caption{RCP and PPE}
        \label{fig:errorplotRCP_huawei}
    \end{subfigure}

    \begin{subfigure}{.4\textwidth}
        \centering
        \includegraphics[width=\linewidth]{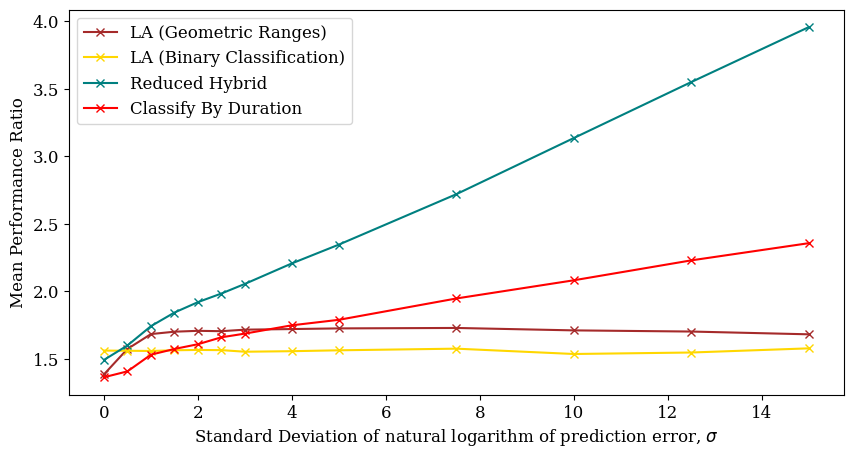}
        \caption{Lifetime Alignment, Classify-By-Duration and Reduced Hybrid}
        \label{fig:errorplotLA_huawei}
    \end{subfigure}
    ~
    \begin{subfigure}{.4\textwidth}
        \centering
        \includegraphics[width=\linewidth]{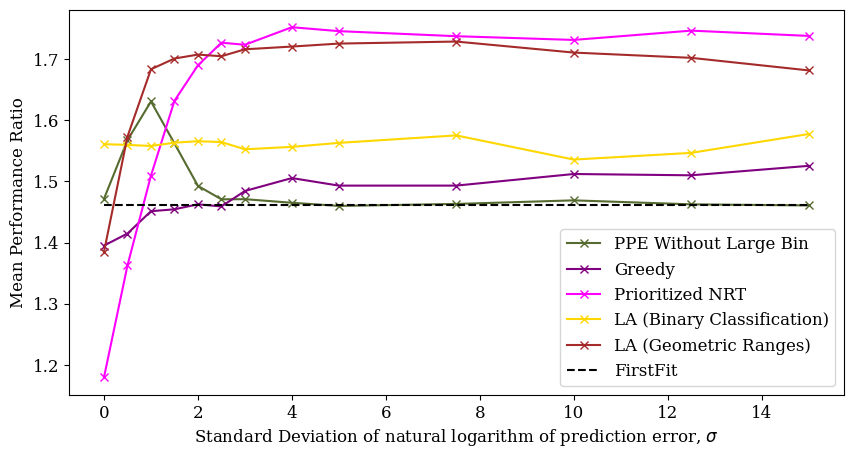}
        \caption{Mean performance ratios of selected algorithms}
        \label{fig:errorplot_huawei}   
    \end{subfigure}

    \caption{Performance of algorithms using the Huawei dataset}
    \label{fig:huawei}
\end{figure}

%% file: content/app-uniform_dist_error.tex
\subsection{Experimentation with Uniform Distribution of Prediction Errors in Learning-augmented Setting} \label{app:uniform}

As mentioned in Section \ref{sec:experimentforlearning}, we have also experimented with a uniform distribution of prediction errors. Assume that predictions have a maximum multiplicative error of $\epsilon$ ($\epsilon \geq 1$). For each item, we first randomly generate a multiplicative prediction error $\delta$ according to a uniform distribution between $1$ and $\epsilon$. Then, we flip an unbiased coin to decide whether the prediction underestimates or overestimates the item duration, i.e., it has equal probabilities to underestimate and overestimate. If it is an underestimation, the predicted duration is set to $\mathrm{Pdur}(r) = \frac{\mathrm{Rdur}(r)}{\delta}$. If it is an overestimation, the predicted duration is set to $\mathrm{Pdur}(r) = \delta \cdot \mathrm{Rdur}(r)$. In either case, it means that the item is predicted to depart at time $I(r)^- + \mathrm{Pdur}(r)$. Note that if $\epsilon = 1$, the predictions perfectly match the real durations of all items. Our evaluation includes a wide range of $\epsilon$ settings: $\epsilon = 1,2,3,4,5,10,20,50,100,1000,10000,100000,1000000$. Figure \ref{fig:errorplot_uni} presents the results.

\begin{figure}[t!]
\begin{subfigure}{0.48\linewidth}
    \centering
    \includegraphics[width=0.9\linewidth]{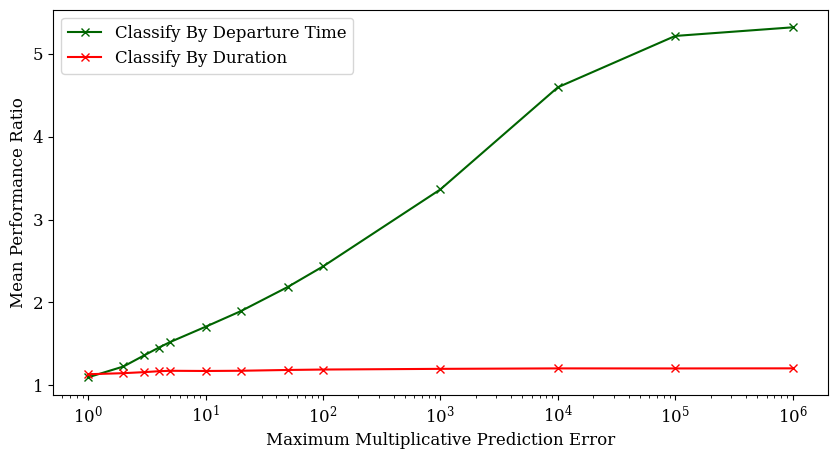}
    \caption{Classify-By-Departure-Time and Classify-By-Duration}
    \label{fig:errorplot_CB_uni}
\end{subfigure}
~
\begin{subfigure}{0.48\linewidth}
    \centering
    \includegraphics[width=0.9\linewidth]{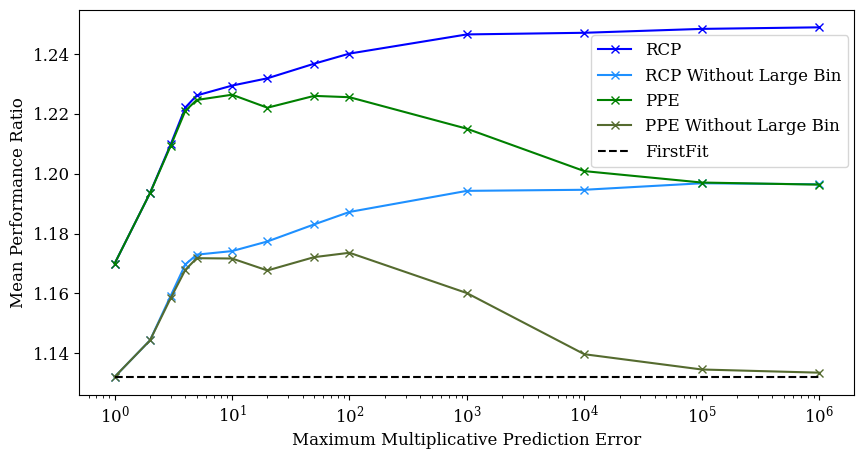}
    \caption{RCP and PPE}
    \label{fig:errorplotRCP_uni}
\end{subfigure}

\begin{subfigure}{0.48\linewidth}
    \centering
    \includegraphics[width=0.9\linewidth]{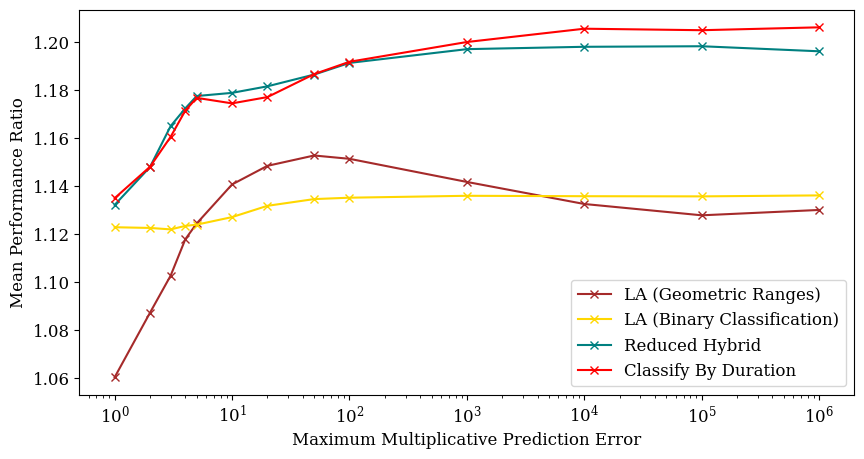}
    \caption{Lifetime Alignment, Classify-By-Duration and Reduced Hybrid}
    \label{fig:errorplotLA_uni}
\end{subfigure}
~
\begin{subfigure}{0.48\linewidth}
    \centering
    \includegraphics[width=0.9\linewidth]{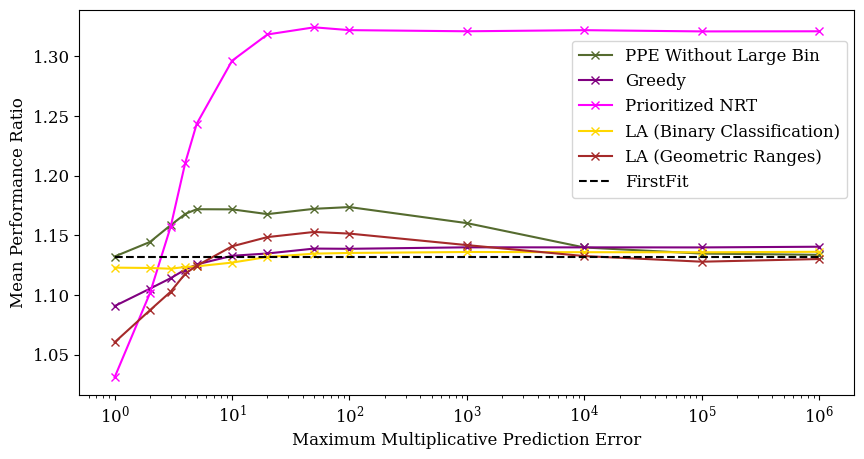}
    \caption{Mean performance ratios of selected algorithms}
    \label{fig:errorplotLA_uni}
\end{subfigure}
    \caption{Learning-Augmented Setting with Uniform Distribution of Prediction Errors}
    \label{fig:errorplot_uni}
\end{figure}